\newif\ifsinglecolumn  %
\singlecolumnfalse
\newif\ifrelaxedv  %
\relaxedvtrue
\newif\ifextendedv  %
\extendedvtrue

\newif\ifmargincomments %
\margincommentsfalse

\ifsinglecolumn
\documentclass[journal,10pt,onecolumn,draftclsnofoot]{IEEEtran}   %
	\relaxedvtrue
\else
\documentclass[10pt,twocolumn,twoside]{IEEEtran}  %
	\relaxedvfalse
\fi

\newif\ifshortest %
\ifextendedv
	\shortestfalse
	\relaxedvtrue
\else
	\shortesttrue
\fi

\IEEEoverridecommandlockouts                              %

\usepackage{mathtools}
\usepackage{graphicx}  %
\usepackage{multicol}  %
\usepackage{multirow}
\usepackage{textcomp}
\usepackage{footmisc}

\usepackage{makeidx}   %
\usepackage{graphicx}  %

\usepackage{url}

\usepackage{amsmath,amssymb}%

\usepackage{amsthm}

\makeatletter
\newcommand{\pushright}[1]{\ifmeasuring@#1\else\omit\hfill$\displaystyle#1$\fi\ignorespaces}
\newcommand{\pushleft}[1]{\ifmeasuring@#1\else\omit$\displaystyle#1$\hfill\fi\ignorespaces}
\makeatother

\usepackage{algorithm}
\usepackage[noend]{algpseudocode}

\newtheorem{theorem}{Theorem}[section]
\newtheorem{definition}[theorem]{Definition}

\newtheorem{lemma}[theorem]{Lemma}

\usepackage[usenames,dvipsnames,svgnames]{xcolor}

\newcommand{\transpose}[1]{\left(#1\right)^T\!\!}

\newcommand{\revIJRR}[1]{#1}
\newcommand{\revTCNS}[1]{#1}
\newcommand{\revTCNSII}[1]{#1}

\newcommand{\jvspace}[1]{\ifrelaxedv\else\vspace{#1}\fi}

\newcommand{\V}{\mathcal V}

\DeclareMathOperator*{\argmin}{arg\,min}

\graphicspath{{./},{./fig/},{./fig/bio/}}

\begin{document}
	\title{\LARGE \bf On the interaction between Autonomous Mobility-on-Demand systems and the power network: models and coordination algorithms}
	\author{Federico Rossi, Ramon Iglesias,
		Mahnoosh Alizadeh and Marco Pavone\IEEEauthorrefmark{1}
		\thanks{Federico Rossi and Marco Pavone are with the Department of Aeronautics and Astronautics, Stanford University, Stanford (CA) 94305. Email: \{frossi2,pavone\}@stanford.edu.} 
		\thanks{Ramon Iglesias is with the Department of Civil and Environmental Engineering, Stanford University, Stanford (CA) 94305. Email: rdit@stanford.edu.}
		\thanks{Mahnoosh Alizadeh is with the 
Electrical \& Computer Engineering Department, University of California, Santa Barbara, Santa Barbara, CA 93106. Email: alizadeh@ucsb.edu.}\thanks{This research was supported by the National Science Foundation under CAREER Award CMMI-1454737 and by the Toyota Research Institute (TRI). This article solely reflects the opinions and conclusions of its authors and not NSF, TRI or any other Toyota entity.}
\thanks{\IEEEauthorrefmark{1} Corresponding author.}
}

\maketitle

\begin{abstract}
We study the interaction between a fleet of electric self-driving
vehicles servicing on-demand transportation requests (referred to as
Autonomous Mobility-on-Demand, or AMoD, system) and the electric power
network. We propose a joint model that captures the coupling between
the two systems stemming from the vehicles' charging requirements,
capturing time-varying customer demand, battery
depreciation, and power transmission constraints. \revTCNS{First, we show that
the model is amenable to efficient optimization. Then, we prove that the
socially optimal solution to the joint problem is a general equilibrium
if locational marginal pricing is used for electricity. Finally, we
show that the equilibrium can be computed by selfish transportation and
generator operators (aided by a non-profit ISO) without sharing private
information. We assess the performance of the approach and its
robustness to stochastic fluctuations in demand through case studies
and agent-based simulations.} Collectively, these results provide a
first-of-a-kind characterization of the interaction between AMoD
systems and the power network, and shed additional light on the
economic and societal value of AMoD.
\end{abstract}
 \jvspace{-2mm}
\jvspace{-4mm}
\section{Introduction}
Private vehicles are major contributors to urban pollution,
 which is estimated to cause over seven million premature deaths worldwide every year. %
Plug-in electric vehicles (EVs) hold promise to significantly reduce urban pollution, both by reducing carbon dioxide emissions from internal-combustion engine vehicles, and by enabling use of renewable and low-polluting power generators as a source of energy for transportation services. %
However, at present, adoption of EVs for private mobility has been significantly hampered by customers' concerns about limited range and availability of charging infrastructure.

The emerging technology of self-driving vehicles might provide a solution to these challenges and thus might represent a key enabler for the widespread adoption of EVs. Specifically, fleets of self-driving vehicles providing on-demand transportation services (referred to as Autonomous Mobility-on-Demand, or AMoD, systems) hold promise to replace personal transportation in large cities by offering high quality of service at lower cost \cite{SpieserTreleavenEtAl2014} with positive effects on safety, parking infrastructure, and congestion.
Crucially, EVs are especially well-suited to AMoD systems. On the one hand, short-range trips typical of urban mobility are well-suited to the current generation of range-limited EVs; on the other hand, intelligent {fleet-wide} policies for rebalancing and charging can ensure that vehicles with an adequate level of charge are available to customers, virtually eliminating ``range anxiety,'' a major barrier to EV adoption. 
To fully realize this vision, however, one needs currently unavailable tools to manage the complex {\em couplings} between AMoD fleet management (e.g., for routing and charging the EVs) and the control of the power network. Specifically, one should consider

\begin{enumerate} 
\item {\em Impact of transportation network on power network}: Concurrent charging of large numbers of EVs can have significant effects both on the stability of the power network and on the local price of electricity (including at the charging stations) \cite{Sioshansi2012,
AlizadehWaiEtAl2016,HadleyTsvetkova2009}. For example, \cite{HadleyTsvetkova2009} shows that in California a 25\% market penetration of (non-autonomous) EVs with fast chargers, in the absence of smart charging algorithms, would increase overall electricity demand in peak load by about 30\%, and electricity prices by almost 200\%.%
\item {\em Impact of power network on transportation network}: Electricity prices can significantly affect travel patterns for EVs. \cite{AlizadehWaiEtAl2016} shows that changes in electricity prices can radically alter the travel patterns and charging schedules of fleets of EVs in a simplified model of the San Francisco Bay Area. This, in turn, would affect electricity prices in a complex feedback loop.
\end{enumerate}

The key idea behind this paper is that, by intelligently routing fleets of autonomous EVs and, in particular, by harnessing the flexibility offered by the routes and schedules for the empty-traveling vehicles, one can {\em actively} control such complex couplings and guarantee high-performance for the overall system (e.g., high passenger throughput, lower electricity costs, and increased integration of renewable energy sources). Additionally, autonomous EVs provide a unique opportunity for joint traffic and energy production management, as they could act as mobile storage devices. That is, when not used for the fulfillment of trip requests, the vehicles could be routed to target charging stations in order to either absorb excess generated energy at time of low power demand (by charging) or inject power in the power network at times of high demand (by discharging), \revTCNS{flattening the  "duck curve" \cite{Denholm2015} and reducing the use of expensive and polluting peaker plants}. 

\emph{Literature review}: 
\revIJRR{Control of AMoD systems has been addressed in multiple lines of work,
including queueing-theoretical approaches \cite{IglesiasRossiEtAl2017},
network flow approaches \ifshortest\cite{RossiZhangEtAl2017}\else\cite{PavoneSmithEtAl2012,RossiZhangEtAl2017}\fi,
integer linear programming and model-predictive control approaches \ifshortest\cite{Alonso-MoraSamaranayakeEtAl2017} \else\cite{ZhangRossiEtAl2016b,Alonso-MoraSamaranayakeEtAl2017}\fi,
and simulation-based approaches \ifextendedv\cite{MaciejewskiBischoff2017,LevinKockelmanEtAl2017,FiedlerCertickyEtAl2018}\fi.
However, throughout these works, AMoD systems are assumed to have no impact on the electric power network. }

The
integration of {\em non-autonomous} EVs within the power network has been addressed in three main lines of work. A first line of work addresses the problem of scheduling charging of EVs (i.e., optimizing the charging profile in \emph{time}) under the assumption that the vehicles' charging schedule has no appreciable effect on the power network \cite{RoteringIlic2011,TusharSaadEtAl2012}. This assumption is also commonly made when selecting the locations of charging stations (i.e., optimizing the charging profile in \emph{space}) \cite{GoekeSchneider2015,PourazarmCassandrasEtAl2016}.
A high penetration of EVs would, however, significantly affect the power network. Thus, a second line of work investigates the effects of widespread adoption of EVs on key aspects such as wholesale prices and reserve margins, for example in macroeconomic \cite{HadleyTsvetkova2009} and game-theoretical \cite{Sioshansi2012,WangLinEtAl2010} settings. %
Accordingly, \cite{AlizadehWaiEtAl2016} investigate joint models for EV routing and power generation/distribution aimed at driving the system toward a socially-optimal solution, \revTCNS{and show that the social optimum can be enforced as a general economic equilibrium}.
Finally, a third line of work investigates the potential of using EVs to regulate the power network and satisfy short-term spikes in power demand. The macroeconomic impact of such schemes (generally referred to as Vehicle-To-Grid, or V2G) has been studied in \cite{KemptonTomic2005a}, where it is shown that widespread adoption of EVs and V2G could foster significantly increased adoption of wind power. Going one step further, \cite{KhodayarWuEtAl2013} proposes a unified model for EV fleets and the power network, and derives a joint dispatching and routing strategy that maximizes social welfare  (i.e., it minimizes the \emph{overall} cost borne by all participants, as opposed to maximizing individual payoffs). However, \cite{KemptonTomic2005a} does not capture the \emph{spatial} component of the power and transportation networks, while \cite{KhodayarWuEtAl2013} assumes that the vehicles' schedules are fixed.

The objective of this paper is to investigate the interaction between AMoD and the electric power network (jointly referred to as Power-in-the-loop AMoD, or P-AMoD, systems) in terms of \revTCNS{modeling, algorithmic, and economic tools} to effectively manage their couplings\ifextendedv~(Figure \ref{fig:overviewProblem})\fi.  
 Our work improves upon the state of the art (in particular, \cite{AlizadehWaiEtAl2016})
along three main dimensions: (i) it provides rigorous models for a fleet of \emph{shared} and \emph{autonomous} EVs; 
(ii) it provides {efficient algorithms} that can scale to large-scale instances; and (iii) it characterizes the vehicles' ability to return power to the power network through vehicle-to-grid (V2G) schemes, and its economic benefits.

\ifextendedv
\begin{figure}[!htb]
 \centering
 \includegraphics[width = .5\textwidth]{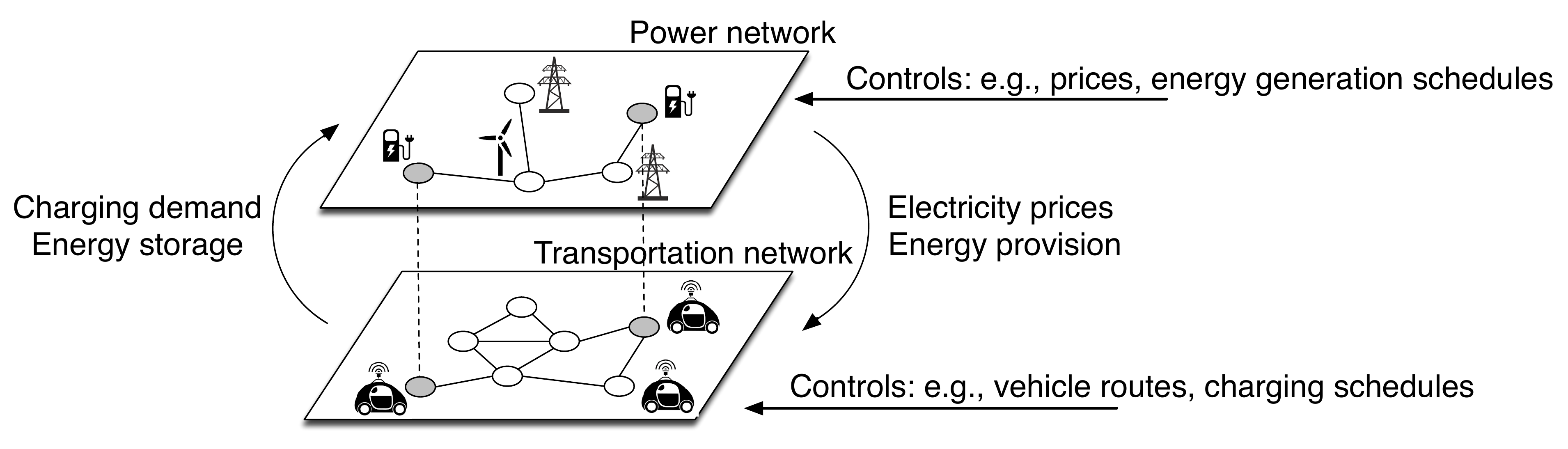}
 \caption{Couplings between an AMoD system and the electric power network. The system-level control of Power-in-the-loop AMoD systems entails the {\em coordinated} selection of routes for the autonomous vehicles, charging schedules, electricity prices, and energy generation schedules, among others.}
 \label{fig:overviewProblem}
\end{figure}
\fi

\emph{Statement of contributions}:  First, we propose a joint model for P-AMoD systems. The model subsumes existing network flow models for AMoD systems and DC models for the power network, and it captures time-varying customer demand and electricity generation costs, congestion in the road network \revTCNSII{(through a simplified threshold model)}, vehicle battery depreciation, power transmission constraints on the transmission lines, and transformer capacity constraints induced by the distribution network. 
\revTCNS{Second, we leverage the model to design tools that optimize the operations of P-AMoD systems and, in particular, maximize social welfare. %
We show that the socially-optimal solution to the P-AMoD problem is a general economic equilibrium under the ubiquitously-used Locational Marginal Pricing electricity pricing scheme - therefore, the social optimum can be realized in the realistic case where transportation and power generator operators are self-interested.
We also propose a distributed privacy-preserving algorithm that the transportation and power network operators can employ to find the equilibrium (specifically, compute the market clearing prices) without disclosing their private information.}
Third, we apply the model and algorithms to a case study of a hypothetical deployment of an AMoD system in Dallas-Fort Worth, TX.
We show that coordination between the AMoD system and the electric power network can have a significant positive impact on the price of electricity (remarkably, the overall electricity expenditure in the presence of the AMoD system can be \emph{lower} than in the case where no vehicles are present, despite the increased demand), while retaining \emph{all} the convenience and sustainability benefits of AMoD. This suggests that the societal value of AMoD systems spans beyond mobility: properly coordinated, AMoD systems can deliver significant benefits to the wider community by helping increase the efficiency of the power network.
\revTCNS{Finally, we show through agent-based mesoscopic simulations that a receding-horizon implementation of the proposed algorithm is highly robust to stochastic fluctuations in demand for transportation and for power. The simulations  
 show that, in absence of coordination, large-scale adoption of electric AMoD can cause widespread blackouts and increase electricity prices by almost 50\%; conversely, the receding-horizon P-AMoD algorithm is able to maintain electricity prices constant, despite the substantial increase in power demand.}

\revIJRR{A preliminary version of this paper was presented at the 2018 Robotics: Science and Systems conference. In this revised and extended version, we provide as additional contributions (i) a rigorous proof that the socially-optimal solution is a general equilibrium, (ii) a privacy-preserving distributed optimization algorithm, \revTCNS{(iii) an extended discussion of the limitations and assumptions of the proposed model,} (iv) additional numerical results, and (v) proofs of all theorems.}

\emph{Organization}: The remainder of this paper is organized as follows. In Section \ref{sec:model} we present a \revTCNS{network flow model for P-AMoD} that captures the interaction between an AMoD system and the power network. %
In Section \ref{sec:equilibrium}, we  show that the socially optimal solution to the P-AMoD problem is a general equilibrium and propose a privacy-preserving distributed optimization algorithm.
In Section \ref{sec:num}, we evaluate our model and algorithm on a case study of Dallas-Fort Worth.
In Section \ref{sec:conclusions},
we draw conclusions and discuss directions for future work. Finally, in the Appendix, we present agent-based simulations and proofs of all theorems.

\jvspace{-1mm}
\section{Model Description and Problem Formulation} \label{sec:model}

We propose a network flow-based model that captures the interaction between an AMoD system and the power network.
The model consists of two parts.

First, we extend the model in \cite{RossiZhangEtAl2017} to a time-varying\revIJRR{, charge-aware} network flow model of an AMoD system with EVs.
We assume that a Transportation Service Operator (TSO) manages the AMoD system in order to fulfill passenger trip requests within a given road network. 
Road links are subject to congestion,
and trip requests arrive according to an {\em exogenous} dynamical process. 
The TSO must not only compute the routes for the autonomous EVs (i.e. \emph{vehicle routing}), but also issue tasks and routes for empty vehicles in order to realign the fleet with the asymmetric distribution of trip demand (i.e. {\em vehicle rebalancing}).
Due to limited battery capacity, the EVs need to periodically charge at charging stations. The price of electricity varies between charging stations -- the charging schedule is determined by the TSO in order to minimize the fleet's operational cost.

The price of electricity itself is a result of the power network operation to balance supply and demand, and varies across the power grid.
Thus, we next review the DC power flow model of the power network and the economic dispatch problem used to calculate market clearing prices for electricity.
The power transmission network comprises spatially-distributed
 energy providers that are connected to spatially-distributed
power network users through high-voltage transmission lines. Transmission capacities (dictated chiefly by thermal considerations) limit the amount of power that can be transferred on each transmission line. 
Load buses are connected to charging stations and other sources of power demand through the distribution systems: these systems induce constraints on the amount of power that can be served to each load bus. 
Power demands other than those from charging stations are regarded as exogenous parameters in this paper.
The power network is controlled by a not-for-profit Independent System Operator (ISO), \revTCNS{which manages a competitive market platform for trading electric energy.} 
The ISO also determines prices at the load buses (and, consequently, at the charging stations) so as to  \revTCNS{achieve market clearing and} guarantee grid reliability while minimizing the overall generation cost (a problem known as \emph{economic dispatch}).

The vehicles' charging introduces a \ifextendedv critical \fi coupling between the transportation and the power networks. The power demands due to charging influence the local price of electricity set by the ISO -- the prices, in turn, affect the optimal charging schedule computed by the TSO. 
Accordingly, we conclude this section by describing the interaction between the two models, and we propose a joint model for Power-in-the-loop AMoD.%

\subsection{Network Flow Model of an AMoD system}
\label{sec:fluidmodel}
We consider a time-varying, finite-horizon model. 
The time horizon of the problem is discretized in $T$ time intervals, \revIJRR{each corresponding to $T_S$ seconds}; the battery charge level of the \ifextendedv autonomous \fi vehicles is similarly discretized in $C$ charge levels, each corresponding to $J_C$ joules.

\textbf{Road network}:
The road network is modeled as a directed graph $R=(\V_R,\mathcal{E}_R)$, where $\V_R$ denotes the node set and $\mathcal{E}_R\subseteq \V_R\times \V_R$ denotes the edge set. 
Nodes $v\in \V_R$ denote either an intersection, a charging station, or a trip origin/destination. Edges $(v,w)\in \mathcal{E}_R$ denote the availability of a road link connecting nodes $v$ and $w$.
For each edge, the length $d_{(v,w)}\in \mathbb{R}_{\geq 0}$ determines the mileage driven along the road link; the traversal time $t_{(v,w)}\in\{1,\ldots,T\}$ characterizes the travel time on the road link in absence of congestion;
the energy requirement $c_{(v,w)}\in\{-C,\ldots,C\}$ models the energy consumption (i.e., the number of charge levels) required to traverse the link in absence of congestion;
and the capacity $\overline f_{v,w} \in \mathbb{R}_{\geq 0}$ captures the maximum vehicle flow rate (i.e., the number of vehicles per unit of time) that the road link can accommodate on top of exogenous traffic without experiencing congestion.

Vehicles traversing the road network can recharge and discharge their batteries at charging stations, whose locations are modeled as a set of nodes $\mathcal{S}\subset \V_R$.
Each charging station $s\in\mathcal{S}$ is characterized by a charging rate $\delta c_s^+\in \{1,\ldots,C\}$, a discharging rate $\delta c_s^-\in\{-C,\ldots,-1\}$, a time-varying charging price $p_{s}^+(t)\in\mathbb{R}$, a time-varying discharging price $p_{s}^-(t)\in\mathbb{R}$, and  vehicle capacity $\overline S_{s}\in \mathbb{N}$.
The charging and discharging rates $\delta c_s^+, \delta c_s^-\in \{1,\ldots,C\}$ correspond to the amount of energy (in charge levels) that the charger can provide to a vehicle (or, conversely, that a vehicle can return to the power grid) in one unit of time. For simplicity, we assume that the charging rates are fixed; however, the model can be extended to accommodate variable charging rates.
The charging and discharging prices $p_{s}^+(t)$ and $p_{s}^-(t)$ capture the cost of one discrete unit charge level (or, conversely, the payment the vehicles receive for returning one unit charge level to the grid) at time $t$; in this paper, we assume that $p_{s}^+(t)=p_{s}^-(t)$ \revIJRR{(in accordance with the assumption of an arbitrage-free market)}.
The vehicle capacity $\overline S_s$ models the maximum number of vehicles that can simultaneously charge or discharge at station $s$. \revIJRR{Charging and discharging (due both to driving activity and to vehicle-to-grid power injection) cause wear in the vehicles' batteries. The battery depreciation per unit charge or discharge is denoted as $d_B$. Battery depreciation captures the cost of replacing a battery at the end of its useful life; note, however, that the vehicle's battery capacity is assumed to remain constant during the model's finite horizon. }

\textbf{Expanded AMoD network}: We are now in a position to rigorously define the network flow model for the AMoD system.
 We introduce an \emph{expanded} AMoD network modeled as a directed graph $G = (\V,\mathcal{E})$.
  The graph $G$ captures the time-varying nature of the problem and tracks the battery charge level of the autonomous vehicles. Specifically, nodes $\mathbf{v}\in\V$ model physical locations at a given time and charge level, while edges $e\in\mathcal{E}$ model road links and charging actions at a given time and charge level. 
Formally, a node $\mathbf{v}\in \V$ corresponds to a tuple $\mathbf{v}=(v_\mathbf{v},t_\mathbf{v},c_\mathbf{v})$, where $v_\mathbf{v}\in \V_R$ is a node in the road network graph $R$; $t_\mathbf{v}\in \{1,\ldots, T\}$ is a discrete time; and $c_\mathbf{v}\in \{1, \ldots, C\}$ is a discrete charge level. The edge set $\mathcal{E}$ is partitioned into two subsets, namely $\mathcal{E}_L$ and $\mathcal{E}_S$, such that $\mathcal{E}_L \cup \mathcal{E}_S = \mathcal{E}$ and $\mathcal{E}_L \cap \mathcal{E}_S = \emptyset$. Edges $e\in\mathcal{E}_L$ represent road links, whereas 
edges $e\in\mathcal{E}_S$ model the charging/discharging process at the stations. %
An edge $(\mathbf{v},\mathbf{w})$ belongs to $\mathcal{E}_L$ when (i) an edge $(v_\mathbf{v},v_\mathbf{w})$ exists in the road network graph edge set $\mathcal{E}_R$, (ii) the link $(v_\mathbf{v},v_\mathbf{w})\in \mathcal{E}_R$ can be traversed in time $t_\mathbf{w}-t_\mathbf{v}=t_{(v_\mathbf{v},v_\mathbf{w})}$, and (iii) the battery charge required to traverse the link is $c_\mathbf{v}-c_\mathbf{w}=c_{(v_\mathbf{v},v_\mathbf{w})}$. Conversely, an edge $(\mathbf{v},\mathbf{w})$ represents a charging/discharging edge in $\mathcal{E}_S$ when (i) $v_\mathbf{v}=v_\mathbf{w}$ is the location of a charging station in $\mathcal{S}$ and (ii) the charging/discharging rate at the charging location $v_\mathbf{v}$ is $(c_\mathbf{w}-c_\mathbf{v})/(t_\mathbf{w}-t_\mathbf{v})=\delta c_{v_\mathbf{v}}^+$ (charging) or $(c_\mathbf{w}-c_\mathbf{v})/(t_\mathbf{w}-t_\mathbf{v})=\delta c_{v_\mathbf{v}}^-$ (discharging). Figure \ref{fig:augmented} (left) shows a graphical depiction of the graph $G$. %

\begin{figure}[!htb]
\jvspace{-2mm}
 \centering
\ifsinglecolumn
 \includegraphics[width=.8\textwidth]{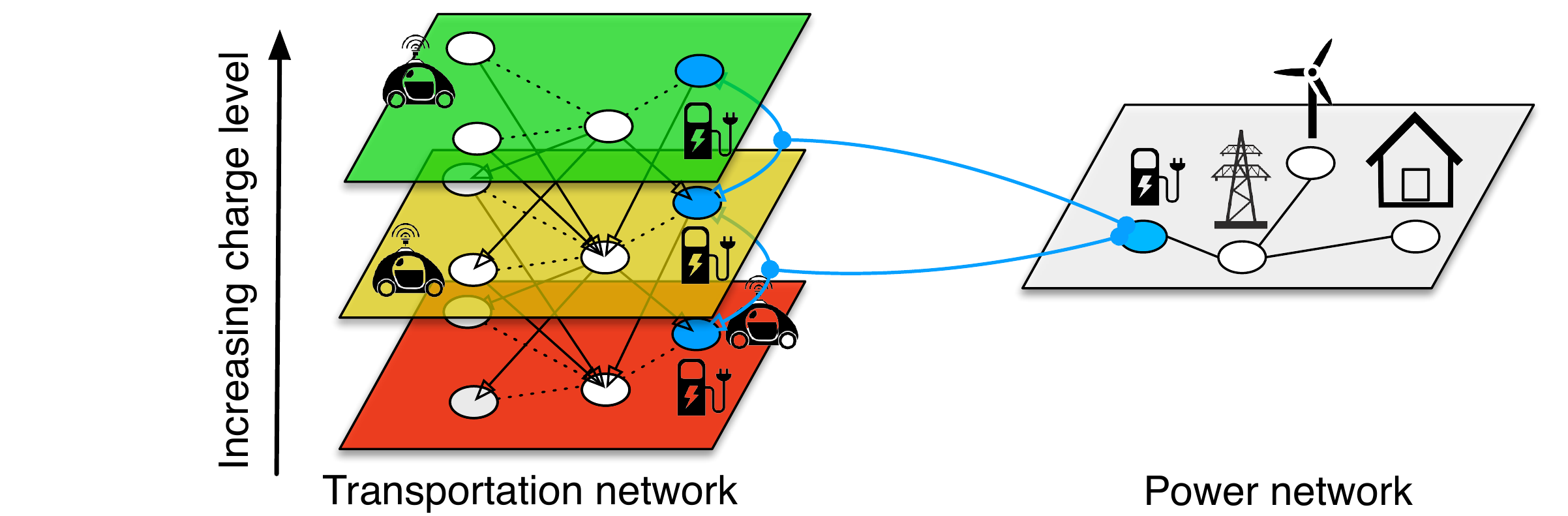}
 \else
 \includegraphics[width=.48\textwidth]{sys_expanded_rev4}
 \fi
 \caption{Augmented transportation and power networks. 
 As vehicles travel on road links (modeled by solid black arrows), their charge level decreases.
 Blue nodes represent charging stations: the flows on charging and discharging edges affect the load at the corresponding nodes in the power network. For simplicity, only one time step is shown.}
 \label{fig:augmented}
\end{figure}
\jvspace{-2mm}

\textbf{Customer and rebalancing routes}:
Transportation requests are represented by the set of tuples $\{(v_m,w_m,t_m,\lambda_m)\}_{m=1}^M$, where $v_m\in \V_R$ is the request's origin location, $w_m\in \V_R$ is the request's destination location, $t_m$ is the requested pickup time, and $\lambda_m$ is the average customer arrival rate (or simply customer rate) of request $m$ within time interval $t_m$.
Transportation requests are assumed to be known and deterministic.

The goal of the TSO is to compute a routing and recharging policy for the self-driving vehicles.
To achieve this, we model vehicle routes as network flows \cite{AhujaMagnantiEtAl1993}.
 Network flows are an \emph{equivalent representation} for routes. Indeed, any route can be represented as a network flow assuming value 1 on edges belonging to the route and 0 elsewhere; conversely, all network flows considered in this paper can be represented as a collection of weighed routes
 \cite[Ch. 3]{AhujaMagnantiEtAl1993}. %

We denote the \emph{customer flow} as the rate of customer-carrying vehicles belonging to a specific transportation request $(v_m,w_m,t_m,\lambda_m)$ traversing an edge $e\in\mathcal{E}$. Formally, for request $m\in\{1, \ldots, M\}$, the customer flow is a function $f_m(\mathbf{v},\mathbf{w}): \mathcal{E}\mapsto \mathbb{R}_{\geq 0}$, that represents the rate of customers belonging to request $m$ traveling from location $v_\mathbf{v}$ to location $v_\mathbf{w}$ (or charging/discharging at location $v_\mathbf{v}=v_\mathbf{w})$ from time $t_\mathbf{v}$ to time $t_\mathbf{w}$, with an initial battery charge of $c_\mathbf{v}$ and a final battery charge of $c_\mathbf{w}$.
Analogously, the rebalancing (or customer-empty) flow $f_0(\mathbf{v},\mathbf{w}): \mathcal{E}\mapsto \mathbb{R}_{\geq 0}$
represents the rate of empty vehicles traversing a road link or charging/discharging. 
Customer flows must satisfy a \emph{continuity} condition: customer-carrying vehicles entering a node at a given time and charge level must exit the same node at the same time and with the same charge level. Equation \eqref{eq:custbal} enforces this condition:

{%
\begin{subequations}
\jvspace{-1mm}
\begin{align}
&\sum_{{\mathbf{u}: (\mathbf{u},\mathbf{v})\in \mathcal{E}}} f_m(\mathbf{u},\mathbf{v}) + 1_{v_\mathbf{v}=v_m}1_{t_\mathbf{v}=t_m} \lambda^{c_\mathbf{v}, \text{in}}_m = \sum_{\mathclap{\mathbf{w}: (\mathbf{v},\mathbf{w})\in \mathcal{E}}} f_m(\mathbf{v},\mathbf{w}) 
\nonumber \\
&\quad+  1_{v_\mathbf{v}=w_m} \lambda^{t_\mathbf{v},c_\mathbf{v}, \text{out}}_m  \,\,\quad  \quad \forall \mathbf{v}\in \mathcal{V}, m\!\in\! \{1,\ldots, M\},  \label{eq:custbal:continuity}\\
&\sum_{c=1}^C \lambda_m^{c,\text{in}} \!=\! \lambda_m, \, \sum_{t=1}^T\sum_{c=1}^C \lambda_m^{t,c,\text{out}} \!=\! \lambda_m  \, \forall m\!\in\! \{1,\ldots,M\},  \label{eq:custbal:intensity}
\end{align}
\label{eq:custbal}
\end{subequations}
}
where the variable $\lambda^{c, \text{in}}_m$ denotes the customer rate departing with charge level $c$ and the variable $\lambda^{t,c, \text{out}}_m $ denotes the customer rate reaching the destination at time $t$ with charge level $c$; both are optimization variables.
Function $1_{x}$ denotes the indicator function of the Boolean variable $x = \{\text{true, false}\}$, that is $1_{x}=1$ if $x$ is true, and $1_x=0$ if $x$ is false. 

Rebalancing flows must satisfy a continuity condition analogous to the one for the customer flows. In addition, rebalancing flows must satisfy a \emph{consistency} condition representing the fact that a customer may only depart the origin location if an empty vehicle is available.
Finally, the initial position and charge level of the vehicles are fixed; the final position and charge level are optimization variables  (possibly subject to constraints, e.g., on the minimum final charge level).
The constraints for the initial and final positions of the rebalancing vehicles at each node $\mathbf{v}\in\V$ are captured by a set of functions $N_I(\mathbf{v})$ and $N_F(\mathbf{v})$, respectively. 
Formally, $N_I(\mathbf{v})$, with $t_\mathbf{v}=0$, denotes the number of rebalancing vehicles entering the AMoD system at location $v_\mathbf{v}$ at time $t_\mathbf{v}$ with charge level $c_\mathbf{v}$.
Conversely, $N_F(\mathbf{v})$, with $t_\mathbf{v}=T$ denotes the number of rebalancing vehicles at location $v_\mathbf{v}$ at time $t_\mathbf{v}$ with charge level $c_\mathbf{v}$.
For $t_\mathbf{v}\neq 0$, $N_I(\mathbf{v})=0$; for $t_\mathbf{v}\neq T$, $N_F(\mathbf{v})=0$.
The overall number of vehicles in the network is $\sum_{\mathbf{v}\in\V} N_I(\mathbf{v})$.
Equation \eqref{eq:rebbal} simultaneously enforces the rebalancing vehicles' continuity condition, consistency condition, and the constraints on the initial and final locations:

{%
\jvspace{-3mm}
\begin{align}
&\sum_{\mathclap{\mathbf{u}: (\mathbf{u},\mathbf{v})\in \mathcal{E}}} f_0(\mathbf{u},\mathbf{v}) \!+\! \sum_{m=1}^M 1_{v_\mathbf{v}=w_m} \lambda^{t_\mathbf{v},c_\mathbf{v}, \text{out}}_m + N_I(\mathbf{v}) = \label{eq:rebbal} \\
& \sum_{\mathclap{\mathclap{\mathbf{w}: (\mathbf{v},\mathbf{w})\in \mathcal{E}}}} f_0(\mathbf{v},\mathbf{w}) \!+\! \sum_{m=1}^M 1_{v_\mathbf{v}=v_m} \!1_{t_\mathbf{v}=t_m} \lambda^{c_\mathbf{v}, \text{in}}_m \!+\! N_F(\mathbf{v}) ,  \forall \mathbf{v}\in \mathcal{V}. \nonumber
\end{align}
\jvspace{-1mm}
} 

\textbf{Congestion}:
We adopt a simple \emph{threshold} model for congestion: the vehicle flow on each road link is constrained to be smaller than the road link's \revTCNS{residual capacity (i.e. the flow of autonomous vehicles that can traverse the link without inducing road congestion, once exogenous vehicle traffic is accounted for)}. \revTCNSII{Equivalently, the traffic speed is assumed to be equal to the free-flow speed whenever the vehicle flow is smaller than the road capacity, and zero whenever the vehicle flow exceeds the road capacity.} The model is analogous to the one adopted in \cite{RossiZhangEtAl2017} and is consistent with classical traffic flow theory \cite{Wardrop1952}.
\ifextendedv In classical flow theory, when the vehicle density on a road link is low, vehicles travel at the free-flow speed, which is approximately constant in this regime \cite{Kerner2009b}. As the vehicle density increases, the vehicle flow achieves an empirically observed maximum (denoted as the road \emph{capacity} in the literature). A further increase in vehicle density causes a dramatic reduction in the traffic speed and the vehicle flow, and signals the onset of congestion. The threshold model constrains traffic to remain in the uncongested regime. \fi
This simplified congestion model is adequate for our goal of computing \emph{control strategies} for the vehicles' routes and charging schedules, and ensures tractability of the resulting optimization problem. Higher-fidelity models can be used for the \emph{analysis} of the AMoD system's operations\revTCNS{: indeed, we employ the high-fidelity Bureau of Public Roads (BPR) congestion model \cite{BPR1964} in the numerical simulations in Appendix \ref{sec:realtime}}. Equation \eqref{eq:congestion} enforces the road congestion constraint:

{%
\jvspace{-1mm}
\begin{equation}
\sum_{c_\mathbf{v}=1}^C \sum_{m=0}^M\! f_m(\mathbf{v},\mathbf{w}) \leq \overline{f}_{(v_\mathbf{v},v_\mathbf{w})},
 \forall (v_\mathbf{v},v_\mathbf{w}) \!\in\! \mathcal{E}_R, t_\mathbf{v}\in \{1,\ldots,T\}.
\label{eq:congestion}
\end{equation}
\jvspace{-3mm}
}

Charging stations can simultaneously accommodate a limited number of vehicles. The station capacity constraint is enforced with Equation \eqref{eq:chargercap}:

{%
\jvspace{-1mm}
\begin{equation}
\sum_{
{\substack{(\mathbf{v},\mathbf{w})\in \mathcal{E}_S:\\ 
v_\mathbf{v}=v_{\mathbf{w}}=v
}}} \sum_{m=0}^M f_m(\mathbf{v},\mathbf{w})  \leq \overline{S}_{v_{\mathbf{v}}}, 
\forall v\in \mathcal{S}, t\in\{1,\ldots, T\}.
\label{eq:chargercap}
\end{equation} 
}

\revTCNS{\textbf{Flow Bundling}:
The goal of the TSO is to select variables
$\{f_m, \lambda^{c, \text{in}}_m, \lambda^{t,c, \text{out}}_m, N_F\}$
so as to minimize the aggregate operational cost borne by AMoD users (which will be formally defined later in this Section).
The size of the edge set $\mathcal{E}$ is $|\mathcal{E}|=O((|\mathcal{E}_R|+|\mathcal{S}|)CT)$ (that is, the asymptotic growth of $|\mathcal{E}|$ is bounded from above by a function  $\overline k(|\mathcal{E}_R|+|\mathcal{S}|)CT $, where $\overline k$ is a positive constants), and the number of customer requests $M$ admits an upper bound $O(|\mathcal{V}_R|^2T)$, since each customer demand is associated with an origin, a destination, and a departure time. The size of the problem is dominated by the customer flow variables in the road network -- the number of such variables is $M|\mathcal{E}|=O((|\mathcal{V}_R|^2T) (|\mathcal{E}_R|+|\mathcal{S}|)CT)$.
Consider a typical problem with 25 road nodes, 200 road links, 30 charge levels, and a horizon of 20 time steps. Such a problem results in a number of variables on the order of $10^{9}$, which can not be solved even by state-of-the-art solvers on modern hardware \cite{Mittelmann2016}.

To overcome this, we propose a \emph{bundling} procedure that allows one to reduce the number of network flows to $O(|\mathcal{V}_R|)$ without loss of information. As a result, the size of the prototypical problem above is reduced to $O(10^6)$ variables, well within the reach of modern solvers. The procedure collects multiple customer demands in a single customer flow, a concept we refer to as \emph{bundled customer flow},
\begin{definition}[Bundled customer flow]
\label{def:bundledflow}
Consider the set of customer requests \linebreak$\{v_m,w_m,t_m,\lambda_m\}_{m=1}^M$. Denote the set of customer destinations as $\mathcal{D}:=\{\cup_{m=1}^M w_m\}$. %
For a given destination $d_B\in \mathcal{D}$, we define a \emph{bundled customer flow} as a function $f_{B,d_B}(\mathbf{u},\mathbf{v}): \mathcal{E}\mapsto \mathbb{R}_{\geq 0}$ that satisfies

 {%
 \jvspace{-1mm}
\begin{subequations}
\begin{align}
&\sum_{\mathclap{\mathbf{u}: (\mathbf{u},\mathbf{v})\in \mathcal{E}}} \; f_{B,d_B}(\mathbf{u},\mathbf{v}) + \sum_{\mathclap{\substack{m\in\{1,\ldots,M\}:\\ w_m=d_B}}}\;1_{v_\mathbf{v}=v_m}1_{t_\mathbf{v}=t_m} \lambda^{c_\mathbf{v}, \text{in}}_m \nonumber \\
&= \sum_{\mathclap{\mathbf{w}: (\mathbf{v},\mathbf{w})\in \mathcal{E}}}\; f_{B,d_B}(\mathbf{v},\mathbf{w}) +\sum_{\mathclap{\substack{m\in\{1,\ldots,M\}:\\w_m=d_B}}}\; 1_{v_\mathbf{v}=w_m} \lambda^{t_\mathbf{v},c_\mathbf{v}, \text{out}}_m, 
\,\, \forall \mathbf{v} \in \V, \label{eq:bundledcustbal_conservation}\\
&\sum_{c=1}^C \lambda_m^{c,\text{in}} \!=\!  \sum_{t=1}^T\sum_{c=1}^C \lambda_m^{t,c,\text{out}} \!=\! \lambda_m, \forall m\in \{1,\ldots,M\}\!:\!w_m\!=\!d_B. \label{eq:bundledcustbal_sumsource}%
\end{align}
\label{eq:bundledcustbal}
\end{subequations}
 \jvspace{-3mm}
} 
\end{definition}

Intuitively, the bundled customer flow for a given destination $d_B$ can be thought of as the sum of customer flows (i.e., network flows satisfying Equation \eqref{eq:custbal}) for \emph{all} customer requests whose destination is node $d_B$.
A bundled customer flow is an \emph{equivalent representation} for a set of customer flows belonging to customer requests sharing the same destination.
The next lemma formalizes this intuition.

\begin{lemma}[Equivalency between customer flows and bundled customer flows]
\label{lemma:bundledflows}
Consider a network $G(\mathcal{V},\mathcal{E})$ and a set of customer requests $\{v_m,w_m,t_m,\lambda_m\}_{m=1}^M$.
Assume there exists a bundled customer flow $\{f_{B,d_B}(\mathbf{u},\mathbf{v})\}_{(\mathbf{u},\mathbf{v})\in\mathcal{E}}$ that satisfies Equation \eqref{eq:bundledcustbal} for a destination $d_B\in\mathcal{D}$. Then, for each customer request $\{v_m,d_B,t_m,\lambda_m\}$ with destination $d_B$, there exists a customer flow $f_m(\mathbf{u},\mathbf{v})$ that satisfies Equation \eqref{eq:custbal}. Furthermore, for each edge $(\mathbf{u},\mathbf{v})\in\mathcal{E}$, $f_{B,d_B}(\mathbf{u},\mathbf{v}) = \sum_{m\in \{1,\ldots,M\} :w_m=d_B}f_m(\mathbf{u},\mathbf{v})$.
\end{lemma}

\emph{Proof sketch:}
The flow decomposition algorithm \cite{AhujaMagnantiEtAl1993} is used to decompose the bundled customer flow into a collection of path flows, each with a single origin node $\mathbf{v}\in\mathcal{V} $ and destination node $\mathbf{w}\in\mathcal{V}$ with $v_\mathbf{w}=d_B$. The customer flow for customer request $(v_m,d_B,t,\lambda)$ is then obtained as the sum of path flows leaving nodes $\{\mathbf{v}=(v_m,t_m,c)\}_{c=1}^C$ with total intensity $\lambda_m$. A rigorous proof is reported in the Appendix.

We can leverage the result in Lemma \ref{lemma:bundledflows} to restate the transportation network model in terms of bundled customer flows, thus dramatically decreasing the model size.
To do so, we note that, according to Lemma \ref{lemma:bundledflows}, Equation \eqref{eq:custbal} is equivalent to Equation \eqref{eq:bundledcustbal}. 
Also, in Equations \eqref{eq:congestion} and \eqref{eq:chargercap}, the quantity $f_m$ only appears as $\sum_m f_m$ and, in accordance with Definition \ref{def:bundledflow} and Lemma \ref{lemma:bundledflows}, $\sum_{m=1}^M f_m = \sum_{d_B} f_{B,d_B} \forall (\mathbf{u}, \mathbf{v})\in\mathcal{E}$. Accordingly, one can replace every occurrence of $\sum_{m=1}^{M} f_m$ with $\sum_{d_B} f_{B,d_B}$ in Equations \eqref{eq:congestion} and \eqref{eq:chargercap} and obtain an equivalent representation of the corresponding constraints.}

\revTCNS{\textbf{Network flow model of an AMoD system}:
The travel time $T_M$ experienced by customers, a proxy for customer welfare, and the overall mileage $D_V$ driven by (both customer-carrying and empty) vehicles, a proxy for vehicle wear, are given by

{%
\jvspace{-1mm}
\begin{align*}
T_M&=\sum_{(\mathbf{v},\mathbf{w})\in \mathcal{E}} t_{\mathbf{v},\mathbf{w}} \sum_{d\in\mathcal{D}}  f_{B,d_b}(\mathbf{v},\mathbf{w}),   \\
 D_V&=\sum_{(\mathbf{v},\mathbf{w})\in \mathcal{E}} d_{v_\mathbf{v},v_\mathbf{w}} \left(f_0(\mathbf{v},\mathbf{w}) + \sum_{d\in\mathcal{D}}  f_{B,d_b}(\mathbf{v},\mathbf{w}) \right).
\end{align*}
}

Note that, for charging edges, $d_{v_\mathbf{v},v_\mathbf{w}}=0$.
The total cost of electricity incurred by the vehicles (including any credit from selling electricity to the power network) is

{%
\vspace{-1mm}
\begin{equation*}
V_E=\sum_{(\mathbf{v},\mathbf{w})\in\mathcal{E}_S} \left(f_0(\mathbf{v},\mathbf{w}) + \sum_{d\in\mathcal{D}}  f_{B,d_b}(\mathbf{v},\mathbf{w}) \right) \delta c_{v_\mathbf{v}} p_{(\mathbf{v},\mathbf{w})},
\end{equation*}
}

where $\delta c_{v_\mathbf{v}}=\delta c_{v_\mathbf{v}}^+$ and $p_{(\mathbf{v},\mathbf{w})}=p_{v_\mathbf{v}}^+$ if $c_\mathbf{w}>c_\mathbf{v}$, $\delta c_{v_\mathbf{v}}=\delta c_{v_\mathbf{v}}^-$ and $p_{(\mathbf{v},\mathbf{w})}=p_{v_\mathbf{v}}^-$ otherwise. 

The overall battery depreciation due to charging and discharging is

{
\jvspace{-1.5mm}
\begin{align*}
V_B=&d_B \left[ \sum_{(\mathbf{v},\mathbf{w})\in\mathcal{E}_S} \left(f_0(\mathbf{v},\mathbf{w}) + \sum_{d\in\mathcal{D}}  f_{B,d_b}(\mathbf{v},\mathbf{w}) \right) |\delta c_{v_\mathbf{v}}| \nonumber \right. \\
& \left.+ \sum_{(\mathbf{v},\mathbf{w})\in\mathcal{E}_L}  \left(f_0(\mathbf{v},\mathbf{w}) + \sum_{d\in\mathcal{D}}  f_{B,d_b}(\mathbf{v},\mathbf{w}) \right) |c_{(v_\mathbf{v},v_\mathbf{w})}| \right].
\end{align*}
\jvspace{-1mm}
}

\revIJRR{(Note that battery depreciation accounts for both charging and discharging, since battery life is determined by the number of charging/discharging cycles incurred by the battery cells). }

The goal of the TSO is to solve the Vehicle Routing and Charging problem, that is, to minimize the aggregate societal cost borne by the AMoD users while satisfying all operational constraints. We define the customers' value of time (i.e., the monetary loss associated with traveling for one time interval) as $V_T$ and the operation cost per kilometer of the vehicles (%
excluding electricity costs) as $V_D$. 
We are now in a position to state the TSO's Vehicle Routing and Charging problem:

{%
\jvspace{-1mm}
\begin{subequations}
\begin{align}
& \underset{f_0, f_{B,d_B}, \lambda^{c, \text{in}}_m, \lambda^{t,c, \text{out}}_m, N_F}{\text{minimize}} 
	& &V_D D_V + V_E + V_B + V_T T_M,\\
	& \text{subject to} &&  \eqref{eq:bundledcustbal}, \eqref{eq:rebbal}, \eqref{eq:congestion}, \text{ and }\eqref{eq:chargercap}.
\end{align}
\label{eq:VRCP}
\end{subequations}
\jvspace{-1mm}
}
}

\revTCNS{The optimization problem in \eqref{eq:VRCP} can be solved with a number of variables on the order of $O((|\mathcal{V}_R|+1) |\mathcal{E}| + MC + |\mathcal{V}_R| C)$. To see this, note that in Equation \eqref{eq:bundledcustbal} the variables $\{ \lambda_m^{t,c,\text{out}}\}_{\{m,t,c\}}$ only appear as part of the sum $ \sum_{{{m\in\{1,\ldots,M\}: w_m=d_B}}}  \lambda^{t,c, \text{out}}_m$ and therefore may be replaced by the smaller set of variables $\{ \lambda_{d_B}^{t,c,\text{out}}\}_{\{d_B,t,c\}}$, where $\lambda_{d_B}^{t,c,\text{out}}:= \sum_{{{m\in\{1,\ldots,M\}: w_m=d_B}}}  \lambda^{t,c, \text{out}}_m$, without loss of generality.
The number of customer flow variables. which dominate the problem size, grows linearly with the number of nodes $|\mathcal{V}_R|$ and does not depend on the time horizon $T$.}

\subsection{\revTCNS{DC model of power network}}
\label{sec:powermodel}

In this paper, the power network is modeled according to the well-known DC model \cite[Ch. 6]{KirschenStrbac2004}, which, by assuming constant voltage magnitudes and determining the power flow on transmission lines solely based on voltage phase angles, represents an approximation to the higher-fidelity AC flow model \cite{GloverSarmaEtAl2011}.
\revTCNS{We remark that any convex optimal power flow model could be readily used in lieu of the DC model, since convex models are also amenable to efficient optimization and can be used to compute Locational Marginal Prices; in this paper, we focus on the DC model as a first step.}
In analogy with the treatment of the AMoD model, we discretize the time horizon of the problem in $T$ time steps. 
The power grid is modeled as an undirected graph $P=(\mathcal{B},\mathcal{E}_P)$, where $\mathcal{B}$ is the node set, commonly referred to as buses in the power engineering literature, and $\mathcal{E}_P\subseteq \mathcal{B} \times \mathcal{B}$ is the edge set, representing the transmission lines.
The subsets of buses representing generators and loads are defined as $\mathcal{G} \subset \mathcal{B}$ and $\mathcal{L} \subset \mathcal{B}$, respectively. %
Generators produce power and deliver it to the network, while loads absorb power from the network. 
Each generator $g\in \mathcal{G}$ is characterized by a maximum output power $\overline p_g(t)$, a minimum output power $\underline p_g(t)$, a unit generation cost $o_g(t)$, and maximum ramp-up and ramp-down rates $p_g^+(t)$ and $p_g^-(t)$, respectively. %
Transmission lines $e\in \mathcal{E}_P$ are characterized by a reactance $x_e$ and a maximum allowable power flow $\overline p_e$ (due chiefly to thermal constraints). The reactance and the maximum allowable power flow do not vary with time.
Each load node $l\in \mathcal{L}$ is characterized by a required power demand $d_l(t)$. The distribution network is not modeled explicitly; however, thermal constraints due to the substation transformers are modeled by an upper bound $\overline{d}_l(t)$ on the power that can be delivered at each load node.

We define 
a generator power function $p:(\mathcal{G},\{1,\ldots,T\})\mapsto \mathbb{R}_{\geq 0}$, and a phase angle function $\theta: (\mathcal B,\{1,\ldots,T\}) \mapsto \mathbb{R}$.
The generation cost is defined as

{%
\jvspace{-1mm}
\[
C_G = \sum_{t=1}^T\sum_{g\in\mathcal{G}} o_g(t) p(g,t).
\]
}
The Economic Dispatch problem entails minimizing the generation cost subject to a set of feasibility constraints \cite{KirschenStrbac2004}:

{%
 \jvspace{-1mm}
\begin{subequations}
\begin{align}
& \underset{p, \theta }{\text{min.}} 
  & &C_G,\\
  & \text{s.t.} && \!\!\!\!\! \sum_{\mathclap{(u,v)\in \mathcal{E}_P}} \frac{\theta(u,t)-\theta(v,t)}{x_{u,v}} + 1_{v\in \mathcal{G}}p(v,t) = 1_{v\in \mathcal{L}}d_v(t)  \nonumber \\
  &&& \!\!\!\! \!\!+\sum_{\mathclap{(v,w)\in \mathcal{E}_P}} \frac{\theta(v,t)-\theta(w,t)}{x_{v,w}}\,, \forall v\in \mathcal{B}, t\in\{1,\ldots,T\}\,, \label{eq:powerbal}\\
&&& \!\!\!\!\! -\overline p_{b_1,b_2}\leq [{\theta(b_1,t)-\theta(b_2,t)}]/{x_{b_1,b_2}} \leq \overline p_{b_1,b_2}\hfill\,, \nonumber \\
&&& \pushright{\forall (b_1,b_2)\in \mathcal{E}_P, t\!\in\! \{1,\ldots,T\}}, \,\,\,\label{eq:powerthermal} \\
&&& \!\!\!\!\! \underline p_g(t) \leq p(g,t)\leq \overline p_g(t), \,\forall g\in \mathcal{G},t\!\in\!\{1,\ldots,T\}, \label{eq:powergenlim}\\
&&& \!\!\!\!\! - p_g^-(t)\leq p(g,t+1) - p(g,t) \leq p_g^+(t)\,,\nonumber \\
&&&  \pushright{\forall g\in \mathcal{G},t\in\{1,\ldots,T-1\}}, \,\,\,\,\,\label{eq:powergenramp}\\
&&&\!\!\!\!\! d_l(t) \leq \overline d_l(t)\,, \quad \quad \quad \quad \forall l\in \mathcal{L},t\in\{1,\ldots,T\}. \label{eq:powerdistlim}
\end{align}
\label{eq:powerdispatch}
\end{subequations}
 \jvspace{-1mm}
}

Equation \eqref{eq:powerbal} enforces power balance at each bus based on the so-called DC power flow equations; Equation \eqref{eq:powerthermal} encodes the transmission lines' thermal constraints; Equation \eqref{eq:powergenlim} encodes the generation capacity constraints; Equation \eqref{eq:powergenramp} encodes the ramp-up and ramp-down constraints; and Equation \eqref{eq:powerdistlim} encodes the thermal constraints of substation transformers.%

\emph{Pricing}:
The unit price of electricity at the load nodes is determined through a mechanism known as Locational Marginal Pricing (LMP) \cite{KirschenStrbac2004}, \revTCNS{ubiquitously used by power network operators in the United States and Western Europe \cite{CheungShamsollahiEtAl1999}.}
The LMP at a node is defined as the \emph{marginal cost} of delivering one unit of power at the node while respecting all the system constraints.
\revIJRR{Accordingly, in this paper, the LMP at each \revTCNSII{load} bus equals the sum of the dual variables (i.e., the shadow prices) corresponding to the power injection constraint \eqref{eq:powerbal} and the substation transformer thermal constraint \eqref{eq:powerdistlim} at the same bus in the Economic Dispatch problem.
}

\subsection{Power-in-the-loop AMoD system}
\label{sec:jointmodel}

The vehicles' charging requirements introduce a \emph{coupling} between the AMoD system and the power network, as shown in Figure \ref{fig:augmented}.
The vehicles' charging schedule produces a load on the power network. Such a load affects the solution to the ISO's Economic Dispatch problem and, as a result, the LMPs. The change in LMPs, in turn, has an effect on the TSO's optimal charging schedule. In absence of coordination, this feedback loop can lead to system instability, as shown for the case of privately-owned, non-autonomous EVs in \cite{AlizadehWaiEtAl2016}. 

In this section, we formulate a \emph{joint model} for the TSO's Vehicle Routing and Charging problem and the ISO's Economic Dispatch problem. 
\revTCNS{We also formulate a cost function that captures the goal of maximizing \emph{social welfare} by minimizing the total cost of mobility (a profit-maximizing formulation would be similar) and the total cost of power generation and transmission.}
\revTCNS{The resulting optimization problem is not directly actionable, since solving it would require the TSO and the ISO to coordinate and share their private information. However, in Section \ref{sec:equilibrium}, we will show that the social optimum is also a general equilibrium if Locational Marginal Pricing is used, and that the operators can compute the equilibrium without exchanging any private information.}%

The coupling between the AMoD model and the electric power model is mediated by the charging stations. A given charging station is represented both by a node $v\in \V_R$ in the road network and by a load node $l\in \mathcal L$ in the power network.
To capture this correspondence, we define an auxiliary function $\mathcal{M}_{\text{P,R}}:\mathcal L\mapsto \{\V_R \cup \emptyset \}$. Given a load node $b\in \mathcal L$, $\mathcal{M}_{P,R}(b)$ denotes the node in $\V_R$ (if any) that represents a charging station connected to $b$. We then define two additional functions, $\mathcal{M}^+_{\text{P,G}}:( \mathcal L,\{1,\ldots,T\})\mapsto \{\mathcal{E}_S \cup \emptyset \}$ and $\mathcal{M}^-_{\text{P,G}}:(\mathcal L,\{1,\ldots,T\})\mapsto \{\mathcal{E}_S \cup \emptyset \}$. The function $\mathcal{M}^+_{\text{P,G}}$ (respectively, $\mathcal{M}^-_{\text{P,G}}$) maps a load node $l$ and a time $t$ to the set of charge (respectively, discharge) edges in $G$ corresponding to station $\mathcal{M}_{\text{P,R}}(l)$ at time $t$. Formally,
 {%
 \begin{align*}
& \mathcal{M}^+_{\text{P,G}}(l,t): \\
&\{  (\mathbf{v},\mathbf{w}) \! \in \! \mathcal{E}_S | 
 v_\mathbf{v}=v_\mathbf{w}, v_\mathbf{v} \! \in \! \mathcal{M}_{\text{P,R}}(l),
 c_\mathbf{v} \! < \! c_\mathbf{w}, t_\mathbf{v}\leq t \! < \! t_\mathbf{w}\},\\
 & \mathcal{M}^-_{\text{P,G}}(l,t): \\
 & \{ (\mathbf{v},\mathbf{w}) \! \in \! \mathcal{E}_S |
  v_\mathbf{v}=v_\mathbf{w}, v_\mathbf{v} \! \in \! \mathcal{M}_{\text{P,R}}(l), c_\mathbf{v} \! > \! c_\mathbf{w}, t_\mathbf{v}\leq t \! < \! t_\mathbf{w}\}.
 \end{align*}
 }
The load at a load bus $l$ can be expressed as the sum of two components: an exogenous demand $d_{l,e}$ \revTCNS{(which includes demand from private, non-autonomous electric vehicles)} and the load due to the chargers connected to that bus, quantitatively,
\revTCNSII{
{%
 \jvspace{-1.5mm}
\begin{align}
\label{eq:buscoupling}
d_l(t)&\!=\!d_{l,e}(t)\!+\!
{J_C \delta c^+_{\mathcal{M}_{\text{P,R}}(l)}}
\sum_{\mathclap{\substack{(\mathbf{v},\mathbf{w})\in\\ \mathcal{M}^+_{\text{P,G}}(l,t)}}} \,\,
\left(f_0(\mathbf{v},\mathbf{w}) \!+\! \sum_{d\in\mathcal{D}}  f_{B,d_b}(\mathbf{v},\mathbf{w}) \right)
\nonumber \\
 & \!\!+ {J_C \delta c^-_{\mathcal{M}_{\text{P,R}}(l)}}
\,\,\sum_{\mathclap{\substack{(\mathbf{v},\mathbf{w})\in \\ \mathcal{M}^-_{\text{P,G}}(l,t)}}}\,\,
 \left(f_0(\mathbf{v},\mathbf{w}) + \sum_{d\in\mathcal{D}}  f_{B,d_b}(\mathbf{v},\mathbf{w})\right)\!\!\!
\end{align} 
}
}
for all $l\in \mathcal{L}, t\in\{1,\ldots,T\}$.

We are now in a position to state the Power-in-the-loop AMoD (P-AMoD) problem:

{%
\jvspace{-2.5mm}
\begin{subequations}
\begin{align}
& \underset{f_0, f_{B,d_B}, \lambda^{c, \text{in}}_m, \lambda^{t,c, \text{out}}_m, N_F,\theta,p }{\text{minimize}} 
	& &  \!\!\!V_T T_{M} + V_D D_v+V_B+C_G,\\
	& \text{subject to} && \!\!\! \eqref{eq:bundledcustbal}, \eqref{eq:rebbal}, \eqref{eq:congestion}, \eqref{eq:chargercap}, \eqref{eq:powerdispatch}, \text{ and }\eqref{eq:buscoupling}. 
\end{align}
\jvspace{-7mm}
\label{eq:P-AMoD}
\end{subequations}
}
\revTCNS{
\subsection{Discussion}
\label{sec:discuss}
Some comments on the modeling assumptions and limitations of the proposed model are in order.

\paragraph{Stochasticity} A key limitation of the network flow modeling approach is that it does not capture stochastic effects, in particular it does not explicitly account for the stochasticity of the customer arrival process, which is assumed to be deterministic and known in advance. Yet, from the mesoscopic perspective of this paper, network flow models are justifiable for three main reasons. First, on the foundational side, previous work by the authors  \cite{IglesiasRossiEtAl2017} has shown that a stochastic queueing network model of an AMoD system, wherein the customer arrival process is Poisson and travel times between stations are stochastic, reduces to a (deterministic) network flow model in the (mesoscopic) limit of large fleet sizes. Notably, in such a limiting regime the network flows represent the expected values of the underlying stochastic quantities. While the extension of the analysis in  \cite{IglesiasRossiEtAl2017} to the P-AMoD setting is beyond the scope of this paper, such a connection suggests network flow models as a principled, first-order \emph{approximation} to higher-fidelity stochastic queueing models. Second, on the control side, network flow models, due to the aforementioned connection to stochastic models and their computationally-favorable (linear) structure, are conducive to the synthesis of effective mesoscopic control policies for transportation systems. Indeed, this is one of the features that has made network flow models one of the most popular tools for mesoscopic control \cite{MahutFlorianEtAl2008, Levin2017}. Third, on the operational side (i.e., at a microscopic control level), stochastic effects in real-time operations can be addressed by leveraging receding-horizon optimization. Indeed, in Appendix \ref{sec:realtime}, we present a receding-horizon implementation of Problem \eqref{eq:P-AMoD}, which incorporates new information on customer demand as it is revealed, and quickly returns solutions amenable to real-time control of P-AMoD systems. Agent-based simulations show that the resulting control policy is \emph{highly robust to stochastic fluctuations in demand for transportation and power}.
}
\revTCNS{\paragraph{Social Welfare}
In order to directly compute and implement a solution to the P-AMoD problem \eqref{eq:P-AMoD}, 
 the TSO and the ISO would have to share the goal of maximizing social welfare and be willing to collaborate on a joint policy.} This assumption is, in general, not realistic: not only do the TSO and ISO have different goals, but they are also generally reluctant to share the information required for successful coordination. However, \revTCNS{in Section \ref{sec:equilibrium},
 we show that the social optimum is a general equilibrium for a self-interested TSO, self-interested power generators, and a non-profit ISO acting as a market broker and using LMP to determine electricity prices. We also propose a  distributed \emph{privacy-preserving} mechanism that an ISO and a TSO can adopt to compute the equilibrium LMPs. Together, these results show that the social optimum can be achieved in the presence of self-interested TSOs and generator operators that wish to minimize their own private cost functions and are unwilling to share private information.}
\paragraph{Ride-sharing}
In this paper we consider single-occupancy vehicles, in line with the mode of operation of current MoD systems. \revTCNS{The extension of the P-AMoD framework to ride-sharing, where multiple passengers share the same vehicle for a portion of their ride, is an interesting avenue for future research.}

\paragraph{DC Model}
The DC model for the power network has some shortcomings, chiefly the inability to handle voltage constraints \cite{Hogan1996} and system-dependent accuracy \cite{StottJardimEtAl2009}. On the other hand, its linearity makes it amenable to large-scale optimization and easy to integrate within the economic theory upon which the transmission-oriented market design is based on \cite{StottJardimEtAl2009}. Moreover, the DC model is widely adopted among ISOs \cite{ONeillDautelEtAl2011}, and its LMP calculations are fairly accurate \cite{OverbyeChengEtAl2004}. 
Hence, the DC model is appropriate for high-level synthesis of joint control policies such as those considered in this paper.

\revTCNS{
\section{A General Economic Equilibrium}
\label{sec:equilibrium}
}
\revTCNS{The social welfare formulation presented in the previous section assumes that the TSO and the ISO both wish to maximize social welfare for given generation costs; also, in order to compute the socially optimal solution to the P-AMoD problem, the TSO and the ISO must be willing to share their private information (e.g., customer transportation requests and power generation costs). In this section, we provide \revTCNS{game-theoretical results} and algorithmic tools to overcome these rather unrealistic assumptions}.

We define a P-AMoD \emph{market} as a perfectly competitive market where self-interested power generators sell power to the power network, a self-interested TSO buys from and sells power to the power network and services transportation requests, and a non-profit ISO acts as a market broker (similar to the model in \cite{WangNegrete-PinceticEtAl2012}).
In this framework, we show that the socially optimal solution to the P-AMoD problem \revTCNS{is a general equilibrium for the TSO and the generators \revTCNSII{(that is, supply and demand of electricity are balanced, and no participant to the market has an incentive to change its policy) \cite{KirschenStrbac2004}} if the ISO sets the price of electricity through Locational Marginal Pricing. Next, we propose a distributed privacy-preserving algorithm that the TSO and the ISO can use to achieve the equilibrium (specifically, compute the equilibrium LMPs) without sharing any information on transportation demand or generation costs.}

\revTCNS{\subsection{The socially optimal solution is a general equilibrium}}

{\begin{theorem}[\revTCNS{The socially optimal solution of the P-AMoD problem is a general equilibrium when Locational Marginal Pricing is used}]
\label{thm:equilibrium}
Consider an optimal solution $\{f_0^\star, f_{B,d_B}^\star, \lambda^{c, \text{in}\star}_m, \lambda^{t,c, \text{out}\star}_m, N_F^\star,\theta^\star,p^\star \}$ to the P-AMoD Problem \eqref{eq:P-AMoD}. Also consider a perfectly competitive market (denoted as the P-AMoD market) where a self-interested TSO solves the Vehicle Routing and Charging problem \eqref{eq:VRCP} by selecting variables $\{f_0, f_{B,d_B}, \lambda^{c, \text{in}}_m, \lambda^{t,c, \text{out}}_m, N_F\}$,
self-interested power generators sell power to the network by determining the revenue-maximizing power generation schedule $\{p\}$, 
and a non-profit ISO acts as a market broker by setting locational marginal prices.
Then  $(\{f_0^\star, f_{B,d_B}^\star, \lambda^{c, \text{in}\star}_m, \lambda^{t,c, \text{out}\star}_m, N_F^\star\},\{p^\star \})$
is a general equilibrium.
\end{theorem}}
\emph{Proof Sketch:} The proof relies on showing that satisfaction of the KKT conditions for Problem \eqref{eq:P-AMoD} implies satisfaction of the KKT conditions for Problem \eqref{eq:VRCP}.
The key insight is that the term $V_E$ in the cost function of Problem \eqref{eq:VRCP} captures the marginal cost imposed by the TSO on the power network, aligning the TSO's incentives with the social optimum. A rigorous proof is reported in the Appendix.

\subsection{A distributed algorithm for the P-AMoD problem}
\label{sec:dual-decomposition}
Next, we show that the TSO and the ISO  can compute the locational marginal prices that enforce the general equilibrium without disclosing their private information.
\revTCNS{
The structural coupling between the transportation and power networks is only mediated by the electricity prices. Exploiting this insight, we use a standard dual decomposition algorithm \cite[Ch. 6.4]{Bertsekas1999} to solve Problem \eqref{eq:P-AMoD} in a distributed manner, similar to \cite{AlizadehWaiEtAl2016}.   
Concretely, the TSO repeatedly solves Problem \eqref{eq:VRCP} with electricity prices proposed by the ISO, and the ISO updates the electricity prices according to the TSO's proposed charging schedule; the procedure is repeated until convergence.
We show that this simple algorithm enjoys two remarkable properties.
First, the TSO and the ISO only exchange publicly-available information (namely, the proposed charging schedule of the AMoD vehicles and the proposed electricity prices); thus, the algorithm is \emph{privacy-preserving}.
Second, at each step, the TSO simply solves Problem \eqref{eq:VRCP}. Thus, a welfare-minded ISO can employ the algorithm to steer a \emph{selfish} TSO towards the social optimum.
It is natural to ask why the ISO would be interested in steering the system towards social welfare. ISOs are non-profits whose charter goal is to match power supply with demand while ensuring grid reliability \cite{ONeillDautelEtAl2011}. As shown in the Appendix and in \cite{AlizadehWaiEtAl2016}, lack of cooperation between the TSO and the ISO can lead to blackouts and to large oscillations in demand: accordingly, steering the TSO towards an equilibrium is well in line with the ISO's goal of ensuring grid reliability. 
}

For ease of notation, we define $f_\diamond = \{f_0 \cup f_{B,d_B}\}$ and we rewrite Equations \revTCNSII{\eqref{eq:bundledcustbal}}-\eqref{eq:rebbal} and \eqref{eq:congestion}-\eqref{eq:chargercap} as, respectively,

{%
\jvspace{-2mm}
\begin{align*}
f_{\text{TSO}}^{\text{eq}}(f_\diamond, \lambda^{c, \text{in}}_m, \lambda^{t,c, \text{out}}_m, N_F)&=0 \text{, (Eq. \revTCNSII{\eqref{eq:bundledcustbal}}-\eqref{eq:rebbal})}, \text{with dual } \lambda^{\text{eq}}_\text{TSO},\\
f_{\text{TSO}}^{\text{ineq}}(f_\diamond, \lambda^{c, \text{in}}_m, \lambda^{t,c, \text{out}}_m, N_F)&\leq 0 \text{, (Eq. \eqref{eq:congestion}-\eqref{eq:chargercap})}, \text{with dual }\mu^{\text{ineq}}_\text{TSO}.
\end{align*}
}
We also rewrite Equations \eqref{eq:powerbal} and \eqref{eq:powerthermal}-\eqref{eq:powerdistlim} as, respectively,

{%
\jvspace{-2mm}
\begin{align*}
f_{\text{ISO}}^{\text{eq}}(f_\diamond,\theta,p)&= 0 \text{, (Eq. \eqref{eq:powerbal})}, &\text{with dual } &\lambda^{\text{eq}}_\text{ISO},\\
\quad f_{\text{ISO}}^{\text{ineq}}(f_\diamond,\theta,p)&\leq 0 \text{, (Eq. \eqref{eq:powerthermal}-\eqref{eq:powerdistlim})}, &\text{with dual } &\mu^{\text{ineq}}_\text{ISO}.
\end{align*}
\jvspace{-1mm}
}

\revTCNSII{The vector $ (\lambda^{\text{eq}}_\text{ISO} + \mu^{\text{ineq}}_\text{ISO})$ denotes the locational marginal price of energy at each bus in the power network and at each corresponding charging node. %
That is,
{%
\begin{equation}
\label{eq:LMP}
p_{(\mathbf{v},\mathbf{w})} = J_C\left( \lambda^{\text{eq}}_\text{ISO}(l_{v_{\mathbf{v}}},t_{\mathbf{v}}) +  \mu^{\text{ineq}}_\text{ISO}(l_{v_{\mathbf{v}}},t_{\mathbf{v}}) \right),
\end{equation}
}
where   $l_{v_{\mathbf{v}}}: v_\mathbf{v}=\mathcal{M}_{\text{P,R}}(l_{v_{\mathbf{v}}})$
and, with a slight abuse of notation, we denote the dual variable corresponding to constraint \eqref{eq:powerdistlim} at edge $(l_{v_{\mathbf{v}}},t_{\mathbf{v}})$ as $\mu^{\text{ineq}}_\text{ISO}(l_{v_{\mathbf{v}}},t_{\mathbf{v}})$.
\ifshortest\else~Note that $p_{(\mathbf{v},\mathbf{w})}$ is the price per discrete energy level, whereas $(\lambda^{\text{eq}}_\text{ISO}(l,t) + \mu^{\text{ineq}}_\text{ISO}(l,t))$ is the price per unit of energy.\fi
}

We consider a partial Lagrangian relaxation of Problem \eqref{eq:P-AMoD}\ifshortest:~\else, that is,\fi

{%
\begin{subequations}
\jvspace{-2mm}
\begin{align}
&\!\!\! \underset{{\substack{f_\diamond, \lambda^{c, \text{in}}_m, \lambda^{t,c, \text{out}}_m, \\ N_F,\theta,p }}}{\text{minimize}} \!\!\!\!\!\!\!\!\!\!\!\!
	& &  V_T T_{M}(f_\diamond) + V_D D_v(f_\diamond)+ V_B(f_\diamond) + C_G(p) \nonumber \\[-1.4em]
&&&  +  \lambda^{\text{eq}}_\text{ISO} f^{\text{eq}}_\text{ISO}(f_\diamond,\theta,p)\!+  \mu^{\text{ineq}}_\text{ISO} f^{\text{ineq}}_\text{ISO}(f_\diamond,\theta,p), \\
	& \text{subject to} && f^{\text{eq}}_\text{TSO}(f_\diamond, \lambda^{c, \text{in}}_m, \lambda^{t,c, \text{out}}_m, N_F)=0,\\
	&&& f^{\text{ineq}}_\text{TSO}(f_\diamond)\leq 0.
\end{align}
\label{eq:P-AMoD_lagrangian}
\jvspace{-2mm}
\end{subequations}
}

The TSO and the ISO iteratively optimize Problem \eqref{eq:P-AMoD_lagrangian} with respect to their own decision variables for a fixed value of the Lagrangian multipliers $\lambda^{\text{eq}}_\text{ISO}$ and $\mu^{\text{ineq}}_\text{ISO}$. Specifically, at step $k$ of the iterative procedure, the TSO solves:

{%
\begin{subequations}
\begin{align}
& \underset{\substack{f_\diamond^k, \lambda^{c, \text{in},k}_m, \lambda^{t,c, \text{out},k}_m, N_F^k}}{\text{minimize}} \!\!\!\!\!\!\!\!\!
	& &  V_T T_{M}(f_\diamond^k) + V_D D_v(f_\diamond^k)+V_B(f_\diamond^k) \label{eq:TSOcost_lagrangian} \\
	&&& + \lambda^{\text{eq},k-1}_\text{ISO} f^{\text{eq}}_\text{ISO}(f_\diamond^k) + \mu^{\text{ineq},k-1}_\text{ISO} f^{\text{ineq}}_\text{ISO}(f_\diamond^k), \nonumber\\
	& \text{subject to} && f^{\text{eq}}_\text{TSO}(f_\diamond^k, \lambda^{c, \text{in},k}_m, \lambda^{t,c, \text{out},k}_m, N_F^k)=0,\\
	&&& f^{\text{ineq}}_\text{TSO}(f_\diamond^k)\leq 0.
\end{align}
\label{eq:VRCP_lagrangian}
\end{subequations}
}
Minimizing the last two terms of Equation \eqref{eq:TSOcost_lagrangian} is equivalent to minimizing the cost of electricity $V_E$ with prices $ \left(\lambda^{\text{eq},k-1}_\text{ISO}+  \mu^{\text{ineq},k-1}_\text{ISO}\right)$. That is,
\[
\argmin_{f_\diamond^k} \lambda^{\text{eq},k-1}_\text{ISO} f^{\text{eq}}_\text{ISO}(f_\diamond^k) + \mu^{\text{ineq},k-1}_\text{ISO} f^{\text{ineq}}_\text{ISO}(f_\diamond^k) = \argmin_{f_\diamond^k} V_E.
\] Thus, Problem \eqref{eq:VRCP_lagrangian} is equivalent to the Vehicle Routing and Charging Problem \eqref{eq:VRCP}.

Analogously, at step $k$, the ISO solves

{\small
\begin{align*}
& \underset{\theta^k,p^k }{\text{minimize}} 
	& & C_G(p^k) +  \lambda^{\text{eq},k-1}_\text{ISO} f^{\text{eq}}_\text{ISO}(\theta^k,p^k) +  \mu^{\text{ineq},k-1}_\text{ISO} f^{\text{ineq}}_\text{ISO}(\theta^k,p^k).
\end{align*}
}
The Lagrangian multipliers are then updated by the ISO as

{%
\begin{align*}
\lambda^{\text{eq},k}_\text{ISO}&=\lambda^{\text{eq},k-1}_\text{ISO}+\alpha_k\left( f^{\text{eq}}_\text{ISO}(f_\diamond^k,\theta^k,p^k) \right),\\
\mu^{\text{ineq},k}_\text{ISO}&=
\max\left(0,\mu^{\text{ineq},k-1}_\text{ISO}+\alpha_k\left( f^{\text{ineq}}_\text{ISO}(f_\diamond^k,\theta^k,p^k) \right)\right),
\end{align*}
}
for an appropriately chosen step size $\alpha_k$%
, and the TSO is informed of the new  proposed price of electricity (i.e., the new value of the sum of the Lagrange multipliers).

Note that the ISO only needs to know the TSO's proposed charging schedule to compute $f^{\text{eq}}_\text{ISO}(f_\diamond^k,\theta^k,p^k)$ and $f^{\text{ineq}}_\text{ISO}(f_\diamond^k,\theta^k,p^k)$; in particular, the TSO does not need to disclose the customers' demand or the planned vehicle routes. Conversely, the ISO only needs to inform the TSO of the proposed price of electricity: the generation costs and the power demands remain private. %

\revTCNS{Convergence of the dual decomposition algorithm for a  ``small enough'' step size $\alpha_k$  follows \ifshortest from \cite[Prop. 6.3.1]{Bertsekas1999}. \else~immediately from \cite[Proposition 6.3.1]{Bertsekas1999}.\fi}

\jvspace{-2mm}
\section{Numerical Experiments}
\label{sec:num}
We study a hypothetical deployment of a P-AMoD system to satisfy medium-distance commuting needs in the Dallas-Fort Worth metroplex, with the primary objective of investigating the interaction between such a system and the Texas power network. Specifically, we study a ten-hour interval corresponding to one commuting cycle, from 5 a.m. to 3 p.m., with 30-minute resolution.
\ifextendedv
Data on commuting patterns is collected from the Census Transportation Planning Products (CTPP) 2006-2010 Census Tract Flows, based on the American Communities Survey (ACS) \cite{FHA2014}. \revIJRR{The AMoD system is assumed to service 30\% of all commuting trips, a scenario capturing low to medium penetration of AMoD}. Departure times are gathered from ACS data \cite{UniStatesCensusBureau2017}.
\fi
~Census tracts in the metroplex are aggregated in 25 districts, as shown in Figure \ref{fig:ClustersAndRoads}. We only consider trips starting and ending in different districts: the total number of customer requests is 400,532.
The commuters' value of time is set equal to \$24.40/hr, in \ifshortest ~line~\else~accordance~\fi~ with DOT guidelines. %
The road network, \ifshortest~\else~the~\fi road capacities, and \ifshortest~\else~the~\fi travel times are obtained from OpenStreetMap data 
\ifextendedv \cite{HaklayWeber2008,Boeing2017} 
\fi
and simplified. The resulting road network, containing 25 nodes and 147 road links, is shown in Figure \ref{fig:ClustersAndRoads}.
\begin{figure}[!htb]
\jvspace{-2mm}
\centering
\ifsinglecolumn
\includegraphics[width=.8\textwidth]{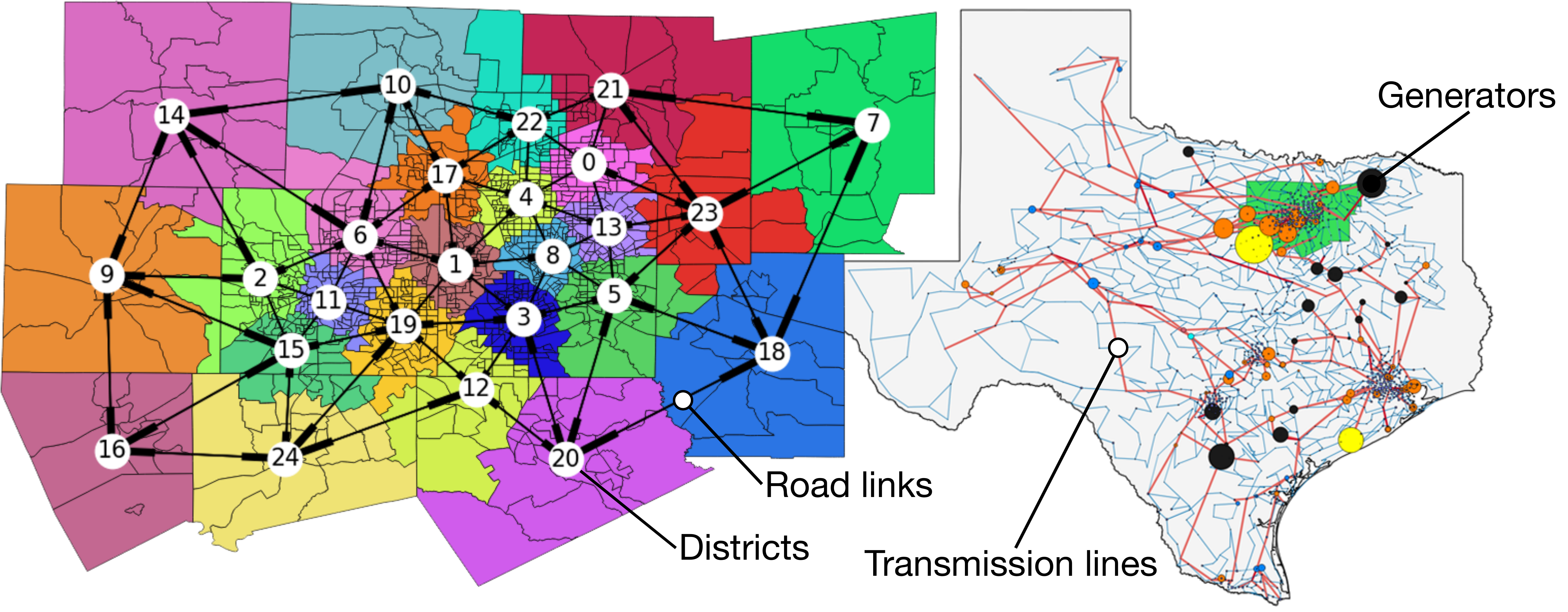}
\else
\includegraphics[width=.5\textwidth]{ClusteredTracts_with_graph_v6}
\fi
\jvspace{-3mm}
\caption{Left: Census tracts and simplified road network for Dallas-Fort Worth. Right: Texas power network model (from \cite{ICSEG2016}). 
\ifextendedv The capacity of each edge equals the overall capacity of roads connecting the start and end district. The travel time between two nodes is the minimal travel time between the centroids of the corresponding districts. \fi}
\label{fig:ClustersAndRoads}
\jvspace{-2mm}
\end{figure}

The battery capacity and power consumption of the EVs are modeled after the 2017 Chevrolet Bolt. %
The cost of operation of the vehicles, excluding electricity costs, is \$0.16/mile (6.55\textcent/mile for maintenance and 9.46\textcent/mile for mileage-based depreciation), in accordance with AAA guidelines. %
The fleet consists of 150,000 vehicles, i.e. 1 AMoD vehicle for every 2.67 customers, similar to the 2.6 ratio in \cite{SpieserTreleavenEtAl2014}.
To represent the possibility that vehicles might not begin the day fully charged, each EV starts the day with a 50\% battery charge and is required to have the same level of charge at the end of the simulation.

We adopt a synthetic model of the Texas power network provided in \cite{ICSEG2016} and portrayed in Figure \ref{fig:ClustersAndRoads}. The model provided does not contain power generation costs: we labeled each generator according to its source of power and assigned generation costs according to U.S. Energy Information Administration estimates\ifextendedv ~\cite{EIA2017}\fi.
The model is also time-invariant; to model the time evolution of power demand and the availability of solar and wind power we used historical data from ERCOT, Texas's ISO\ifextendedv { \cite{ERCOT2017}}\fi, and we imposed ramp-up and ramp-down constraints of 10\%/hr and 40\%/hr on the generation capability of nuclear and coal power plants, respectively.

{
\begin{table*}[h!]
\centering
\caption{Simulation results (one commuting cycle, 10 hours).}
	\begin{tabular}{r|r|rr|rr|rr}
	 \multicolumn{1}{c}{} & \multicolumn{1}{c}{} &  \multicolumn{2}{c}{\$15,734 battery} & \multicolumn{2}{c}{\$1,573 battery} &  \multicolumn{2}{c}{No depreciation} \\
	 \multicolumn{1}{c}{}& Baseline &  P-AMoD & Uncoord. &  P-AMoD & Uncoord. &  P-AMoD & Uncoord. \\ 
	\hline
	 Avg. customer travel time [h] & - &  1.0277 &  1.0277 &  1.0277 &  1.0277 &  1.0277 &  1.0277\\
	Total energy demand [GWh] &517.498 & 520.543 & 520.543& 520.543 & 520.544 & 520.590 & 520.966\\
	Total electricity expenditure [k\$] & 39,617.36 & 39,847.18 & 39,865.34  & 39,847.22 & 40,552.90 & 39,488.93 & 39,519.98\\
\quad w.r.t. baseline [k\$]&  & +229.82 & +247.98 &  +229.83 & +935.54 & -128.43 & -97.38\\ 
	Avg. price in DFW [\$/MW] & 78.75 & 78.68 & 78.79  & 78.69 & 82.23  & 76.89 & 77.12 \\
	TSO  electricity expenditure [k\$] & -& 228.86 & 237.04 & 228.90 & 258.36 & 228.55 & 408.18 \\
	\end{tabular}
\label{tab:results}
\jvspace{-2em}
\end{table*}
}

We  compare the results of three simulation studies. In the baseline simulation study, no electric vehicles are present: we consider the power network \emph{in isolation} subject only to exogenous loads. In the P-AMoD simulation study, we solve Problem \eqref{eq:P-AMoD}, which embodies the cooperation between the TSO and the ISO \revTCNS{and corresponds to the equilibrium in Theorem \ref{thm:equilibrium}}. Finally, in the uncoordinated simulation study, we first solve the TSO's Vehicle Routing and Charging problem with \emph{fixed} electricity prices obtained from the baseline simulation study; we then compute the load on the power network resulting from the vehicles' charging and discharging, and solve the ISO's Economic Dispatch problem with the updated loads.
The uncoordinated simulation study captures the scenario where the TSO attempts to minimize its passengers' cost while disregarding the coupling with the power network.

\revIJRR{For each study, we consider three different levels of battery depreciation. In the first case, the battery replacement cost is \$15,734 (corresponding to the list price of a Chevrolet Bolt battery)
 and vehicles' batteries are fully depreciated over  1,000 charge-discharge cycles, in line with the performance of current battery technology. In the second case,  the battery replacement cost is set to one tenth of the current one (or, equivalently, the vehicles' battery life is 10,000 cycles). In the third case, battery depreciation is neglected.}

Table \ref{tab:results} and Figure \ref{fig:LMPs_TX} show the results.
The quality of service experienced by TSO customers, measured by the average travel time, is virtually identical in the P-AMoD and in the uncoordinated case. The energy demand of the AMoD system is also very similar in both cases. On the other hand, the effect of coordination on the overall electricity expenditure is noticeable. Specifically, with current battery technology, coordination causes a 7.3\% reduction in the TSO's electricity expenditure compared to the uncoordinated case, corresponding to savings of \$9M per year (assuming two commuting cycles per day and 250 work days per year). As battery prices are reduced ten-fold, the urgency of coordination between AMoD systems and the power network increases. In absence of coordination, the TSO's attempts to greedily charge and return power to the grid backfire, resulting in a four-fold increase in the TSO's electricity bill, a  4.4\% increase in the unit price of electricity in the Dallas-Fort Worth area, and an additional expenditure of \$935k per day, or \$467M  per year, in electricity costs borne by all power network customers. 
 Conversely, coordination between the TSO and the ISO ensures that the unit price of electricity in the Dallas-Fort Worth area remains the \emph{same} as in the baseline case, and results in savings of \$14.7M/year for the TSO compared to the uncoordinated case.  A further reduction in the replacement cost of the batteries allows coordination between the AMoD system and the power network to reduce the \emph{total} expenditure for  electricity by \$128k per commuting cycle (\$64M/year) compared to the baseline case, despite the increased demand.
 In other words, a P-AMoD system allows a TSO to deliver on-demand transportation without an increase in overall electricity expenditure -- a remarkable, and perhaps surprising, finding.
 In the uncoordinated case, the presence of the TSO also reduces the overall electricity expenditure by \$97k/cycle compared to the baseline case - however, the reduction is offset by a \$180k/cycle increase in the TSO's own electricity bill  compared to the coordinated case.

 \revTCNS{Collectively, these results show that, even with current battery technology, the savings that can be achieved through coordination between AMoD systems and the power network are highly significant; future battery technology could unlock additional savings of hundreds of millions of dollars and reduce the overall electricity expenditure by tens of millions of dollars per year.}

Who benefits from the reduction in energy expenditure? From the last two rows in Table \ref{tab:results}, one can see that, in the case where no depreciation is considered, the average price of electricity in the P-AMoD case is 2.37\% lower than in the uncoordinated case in Dallas-Fort Worth (corresponding to savings %
 of \$ 147M/year for Dallas-Fort Worth power network customers, excluding the TSO). The energy expenditure of the TSO in the P-AMoD case is 44\% lower than in the uncoordinated case (a saving of \$180k per commuting cycle, corresponding to close to \$90M/year). Finally, electricity customers outside of Dallas experience a small reduction of 0.23\% in their energy expenditure.
 Thus, the majority of the benefits of coordination are reaped by customers of the power network in the region where the AMoD system is deployed; the TSO also benefits from a noticeable reduction in its electricity expenditure.
\begin{figure}[!htb]
\centering
\jvspace{-1em}
\ifsinglecolumn
\includegraphics[width=.8\textwidth]{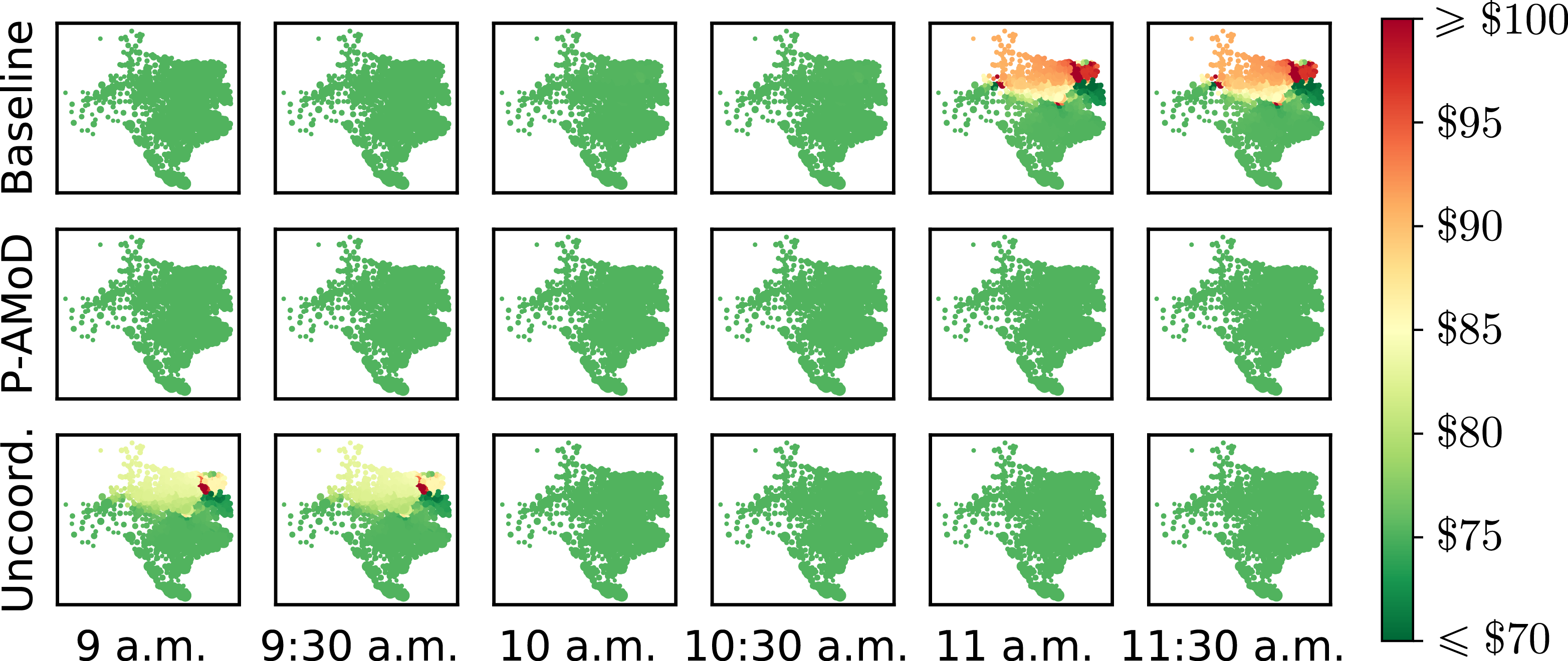} 
\else
\includegraphics[width=.49\textwidth]{LMPs_TX5_compact.png} 
\fi
\jvspace{-3mm}
\caption{LMPs in Texas between 9 a.m. and 11:30 a.m.
The presence of the AMoD fleet can reduce locational marginal prices; coordination between the TSO and the ISO can yield a further reduction. A battery replacement cost of \$1,573 is considered.
}
\label{fig:LMPs_TX}
\jvspace{-2mm}
\end{figure}
Figure \ref{fig:LMPs_TX} shows this phenomenon in detail \revIJRR{for the scenario where the battery replacement cost is \$1,573}. The presence of the AMoD system results in a decrease in the LMPs with respect to the baseline case (11-11:30 a.m.). As electricity prices increase, empty vehicles travel to carefully chosen stations to sell their stored energy back to the network: this results in reduced congestion and lower prices in the power network, even in the absence of coordination. Crucially, coordination between the TSO and the ISO can result in further decreases in the price of electricity with respect to the uncoordinated case (9-9:30 a.m.), significantly curtailing the impact of the AMoD system on the power network.
By leveraging their battery capacities and acting as mobile storage units, the EVs are able to reduce congestion in the power transmission network: this results in lower LMPs in the Dallas-Fort Worth region, and hence lower electricity expenditure.

Simulations were carried out on commodity hardware (Intel Core i7-5960, 64 GB RAM) and used the MOSEK LP solver.
The source code is available online\footnote{\url{https://dx.doi.org/10.5281/zenodo.3241651}}\ifextendedv ~under an open-source license.~\fi %
The simulations required 3,923s for the P-AMoD scenario, 2,885s for the uncoordinated scenario, and 4.55s for the baseline scenario.
While such computation times could be improved by using high-performance computational hardware, in Appendix \ref{sec:realtime} we present a receding-horizon algorithm for P-AMoD which, in addition to the intrinsic robustness benefits of closed-loop control, can be solved in  minutes on commodity hardware
 and returns integral solutions that are directly amenable to control of P-AMoD systems. \revTCNS{The algorithm allows us to perform agent-based simulations that provide further insights into the value of P-AMoD and showcase the robustness of the proposed approach to stochastic fluctuations in customer demand.
 }

\jvspace{-3.5mm}
\section{Conclusions and Future Work} \label{sec:conclusions}
\jvspace{-.4mm}

In this paper we studied the interaction between an AMoD system and the electric power network. \revTCNS{The network flow model} we proposed subsumes earlier models for AMoD systems and for the power network; critically, it captures the coupling between the two systems and allows for their \emph{joint optimization}. 
\revTCNS{We showed that the jointly optimal solution to the P-AMoD problem is a general economic equilibrium, and we proposed a distributed privacy-preserving algorithm that allows agents to find the equilibrium without sharing private information about customer requests, generation costs, or power demands: thus, the results in this paper are applicable to the realistic case where the TSO and generator operators are self-interested.} We applied our model and algorithms to a case study of an AMoD deployment in Dallas-Fort Worth, TX. The case study showed that, depending on the maturity and cost of battery technology, coordination between the TSO and the ISO can result in a \emph{reduction} in the overall electricity expenditure (despite the increase in demand), while having a negligible impact on the TSO's quality of service; conversely, lack of coordination can result in large increases in power prices for power network customer and TSOs alike.
\revIJRR{These results are corroborated by agent-based simulations
\ifextendedv presented \else\fi
in the Appendix.}
 
This work opens multiple avenues of research. %
First, we plan to  capture the impact of cooperation between the TSO and the ISO on the power \emph{distribution} network by incorporating   convex optimal power flow models.
\revTCNS{Second, we plan to develop a \emph{stochastic} (queueing-theoretical) model of P-AMoD, which explicitly captures the stochastic nature of demand for transportation and power, and enables the design of controllers that directly mitigate large-scale stochastic fluctuations.}
\revTCNS{Third, we will extend our model to capture the scenario where multiple TSOs compete for customers while sharing the same transportation and power infrastructure, extending our previous results in \cite{AlizadehWaiEtAl2017}.}
\revIJRR{Fourth, we will extend the P-AMoD model to capture other modes of provision of service, including heterogeneous fleets where vehicles may differ in size, seating capacity, and battery capacity, and ride-sharing mechanisms where multiple customers with similar origins and destinations can travel in the same vehicle.} 
Fifth, the model of the power network considered in this paper does not capture ancillary services such as regulation and spinning reserves. We will extend our model to capture those and evaluate the feasibility of using coordinated fleets of EVs to aid in short-term control of the power network.
Finally, we wish to explore the effect of TSO-ISO coordination on penetration of renewable energy sources, and to determine whether large-scale deployment of AMoD systems can increase the fraction of renewable power sources in the generation power mix.

\jvspace{-3mm}
\bibliographystyle{IEEEtran} 
\bibliography{../../../bib/main,../../../bib/ASL_papers}

\jvspace{-5mm}
\begin{appendix}
\subsection{Agent-based simulations of P-AMoD}
\label{sec:realtime}
\ifextendedv
In this appendix we present agent-based simulations to further explore the impact of P-AMoD on the electric power network \revTCNS{and assess the impact of stochasticity on the performance of the proposed approach}.

\subsubsection{A Receding-Horizon Algorithm for P-AMoD}
~First, by leveraging the structural insights from the network flow optimization problem, along with a few mild assumptions, we devise a computationally efficient control algorithm that solves the \mbox{P-AMoD} Problem \eqref{eq:P-AMoD} in a receding-horizon fashion.

To reduce the computational complexity of the optimization problem, we \emph{decouple} the customer routing process from the P-AMoD optimization.
The key assumption is that customer-carrying trips follow pre-computed routes and are never interrupted by a charging/discharging event. 
Formally, customer trips from node $i\in\V_R$ to node $j\in\V_R$ follow a fixed route with a travel time of  $t_{i\rightarrow j}$ and a required charge of $c_{i\rightarrow j}$. 
Thus, customer flows $\{f_{B,d_B}(\mathbf{u},\mathbf{v})\}_{(\mathbf{u},\mathbf{v}),d_B}$ are no longer part of the optimization variables and Equation \eqref{eq:bundledcustbal_conservation} is redundant. However, the initial and final charge of the customer-carrying vehicles $\{\lambda^{c,\text{in}}_m\}$ and $\{\lambda^{t,c,\text{out}}_m\}$ remain optimization variables.
The following constraint ensures that charge is conserved along customer routes, that is, that vehicles traveling from $i$ to $j$ and departing at time $t$ at charge level $c$  arrive at time $t+t_{i\rightarrow j}$ with charge $c-c_{i\rightarrow j}$:

{
\jvspace{-1mm}
\begin{align}
&\lambda_m^{t,c,\text{out}} = 
\begin{cases}
\lambda_m^{c+c_{v_m\rightarrow w_m},\text{in}} &\text{if } t_{m}=t-t_{v_m\rightarrow w_m}\\
0 & \text{otherwise}
\end{cases} \label{eq:custbalfast_conservation}
\end{align}}
for all $t\in \{1, \ldots, T\}, c\in\{1, \ldots, C\}, m\in\{1, \ldots, M\}$. 
The cost function is also modified to remove the customers' travel times, and road congestion constraints are adjusted to account for the traffic induced by customer-carrying vehicles.
Specifically, the congestion induced by customer-carrying vehicles is \emph{fixed} for given customer demand, since customers follow pre-defined routes. We denote the residual capacity of a road link $(v_\mathbf{v},v_\mathbf{w})\in  \mathcal{E}_R$ at time $t$ (i.e., the capacity of the road link once AMoD customer-carrying trips are accounted for) as $\overline{f}_{(v_\mathbf{v},v_\mathbf{w}),t}$.
The congestion constraints on road links \eqref{eq:congestion} become
\begin{align}
&\sum_{c_\mathbf{v}=1}^C f_0(\mathbf{v},\mathbf{w})\leq \overline{f}_{(v_\mathbf{v},v_\mathbf{w}),t_\mathbf{v}}, \nonumber \\
 &\quad \quad \quad \quad \forall (v_\mathbf{v},v_\mathbf{w}) \in \mathcal{E}_R, t_\mathbf{v}\in \{1,\ldots,T\}.
\label{eq:congestionfast}
\end{align}

Only rebalancing vehicles traverse the charging and discharging links: thus, the capacity constraint of the charging stations  \eqref{eq:chargercap} and the coupling equation \eqref{eq:buscoupling} are rewritten as

{
\begin{align}
 \sum_{
{\substack{(\mathbf{v},\mathbf{w})\in \mathcal{E}_S:\\ 
v_\mathbf{v}=v_{\mathbf{w}}=v
}}}  f_0(\mathbf{v},\mathbf{w})  &\leq \overline{S}_{v_{\mathbf{v}}}, 
&\forall v\in \mathcal{S}, t\in\{1,\ldots, T\},
\label{eq:chargercapfast}
\end{align}
}

{
\begin{align}
\label{eq:buscouplingfast}
d_l(t)=&d_{l,e}(t)+J_C\delta c^+_{\mathcal{M}_{\text{P,R}}(l)} 
\sum_{\substack{(\mathbf{v},\mathbf{w})\in\\ M^+_{P,G}(l,t)}} 
 f_0(\mathbf{v},\mathbf{w})  \nonumber \\
 &  + J_C \delta c^-_{\mathcal{M}_{\text{P,R}}(l)} 
\sum_{\substack{(\mathbf{v},\mathbf{w})\in\\ M^-_{P,G}(l,t)}} f_0(\mathbf{v},\mathbf{w}), \nonumber \\
  & \quad \quad\quad\quad   \forall l\in \mathcal{L}, t\in\{1,\ldots,T\}.
\end{align} 
}

In order to adapt the problem  for use in a receding-horizon implementation, several further modifications are required. Specifically,
\begin{itemize}
\item{\emph{Outstanding customers:}}
In Problem \eqref{eq:P-AMoD}, future customer demand is assumed to be perfectly known -- conversely, in a real-time implementation, unforeseen transportation requests may be originated at any time. As a result, some customers may not be assigned to a vehicle when they arrive at their departure node. We denote the set of such waiting customers as  \emph{outstanding requests}.
Outstanding requests are assumed to wait at the departure station until a vehicle is available. 
The departure time of the outstanding requests is  an optimization variable; the goal is to service such requests as quickly as possible.

Formally, outstanding requests are characterized as the set of $M^o$ tuples $\{(v_{m^o},w_{m^o},\lambda_{m^o})\}_{m^o=1}^{M^o}$, where $v_{m^o}\in V_R$ is the outstanding request's origin location, $w_{m^o}\in\V_R$ is the outstanding request's destination location, and $\lambda_{m^o}$ is the 
average arrival rate (i.e., the number of outstanding customers divided by the duration of one time step). %
For each outstanding request $m^o\in\{1,\ldots,M^o\}$, the set of variables $\{\lambda^{t,c,  \text{in}}_{m^o}\}_{c,t}$ denotes the number of customers  per unit time departing at time $t$ at charge level $c$; in analogy with customer requests, the set of variables $\{\lambda^{t,c, \text{out}}_{m^o}\}_{t,c} $ denotes the number of customers per unit time reaching the destination at time $t$ with charge level $c$. Both are optimization variables. The following constraints ensure that outstanding requests are serviced within the optimization horizon, in analogy with Equations \eqref{eq:bundledcustbal_sumsource} and \eqref{eq:custbalfast_conservation} for regular customers:
\begin{subequations}
\begin{align}
&\sum_{t=1}^T \sum_{c=1}^C \lambda_{m^o}^{t,c,\text{in}} = \lambda_{m^o}, \quad \quad \forall {m^o}\in \{1,\ldots,M^o\}, \label{eq:outstandingcustbal_sumsource}\\
&\lambda_{m^o}^{t,c,\text{out}} = 
\lambda_{m^o}^{t-t_{v_{m^o}\rightarrow w_{m^o}},c+c_{v_{m^o}\rightarrow w_{m^o}},\text{in}}  \label{eq:outstandingcustbal_conservation} \\
& \quad \quad  \forall t\in \{1, \ldots, T\}, c\in\{1, \ldots, C\}, {m^o}\in\{1, \ldots, M^o\} \nonumber.
\end{align}
\label{eq:outstandingcustbal}
\end{subequations}

The overall wait time for outstanding customers can then be computed as
\[
T_M^o = \sum_{{m^o}=1}^{M^o} t \sum_{c=1}^C \lambda_{m^o}^{t,c,\text{in}}.
\]

\item{\emph{Vehicle end charge:}}
In order to achieve satisfactory closed-loop performance and to trade off between servicing present demand and ensuring vehicles are available for future customers,  the final charge level of rebalancing vehicles is constrained to be higher than a given threshold $\underline{C}^T$:
\begin{equation}
f_0(\mathbf{v},\mathbf{w})=0 \quad \forall (\mathbf{v},\mathbf{w})\in\mathcal{E}: c_\mathbf{w}\leq \underline{C}^T, t_\mathbf{w}=T \label{eq:realtime-minendcharge}
\end{equation}

\item{\emph{Feasibility:}}
Problem \eqref{eq:P-AMoD} is not guaranteed to admit a solution for arbitrary transportation requests and arbitrary numbers of vehicles. To ensure persistent feasibility of the receding-horizon controller, slack variables (associated with a high cost) are introduced in Equations \eqref{eq:bundledcustbal_sumsource} and \eqref{eq:outstandingcustbal_sumsource}, allowing customer requests to be dropped to preserve feasibility. As a result, so long as the Economic Dispatch problem is feasible, the P-AMoD problem always admits a feasible solution where no customers are transported and no vehicle moves, charges, or discharges, ensuring persistent feasibility.
\end{itemize}

\paragraph{A receding-horizon controller}
We are now in a position to present the receding-horizon P-AMoD problem. We denote the distance traveled by the rebalancing vehicles as
\[
 D_V^0=\sum_{(\mathbf{v},\mathbf{w})\in \mathcal{E}} d_{v_\mathbf{v},v_\mathbf{w}} f_0(\mathbf{v},\mathbf{w}),
\]
and the depreciation of the rebalancing vehicles' batteries as
\[
V_B^0 = V_B=\sum_{(\mathbf{v},\mathbf{w})\in\mathcal{E}_S} f_0(\mathbf{v},\mathbf{w}) |\delta c_{v_\mathbf{v}}| d_B.
\]

 We pose the receding-horizon  P-AMoD problem as
\begin{align}
\label{eq:P-AMoD-RT}
\,\,\,\, & \underset{\mathclap{\substack{f_0, \lambda^{c, \text{in}}_m, \lambda^{t,c, \text{out}}_m,\\\lambda^{t,c, \text{in}}_{m^o}, \lambda^{t,c, \text{out}}_{m^o}, N_F,\theta,p}} }{\text{minimize}} 
	& & T^o_M + V_D D_V^0+ V_B^0+ C_G\\%
	& \text{subject to} && \eqref{eq:rebbal}, \eqref{eq:bundledcustbal_sumsource}, \eqref{eq:powerdispatch}, \eqref{eq:custbalfast_conservation}, \eqref{eq:congestionfast}, \nonumber \\
	&&&\eqref{eq:chargercapfast}, \eqref{eq:buscouplingfast}, \eqref{eq:outstandingcustbal}, \text{ and } \eqref{eq:realtime-minendcharge} \nonumber
\end{align}

Problem \eqref{eq:P-AMoD-RT} has $O(|\mathcal{E}| + MC + |\mathcal{V}_R| C + T(|\mathcal{G}|+|\mathcal{E}_p| + |\mathcal{B}|))$variables: compared to Problem \eqref{eq:P-AMoD}, the problem size does not depend on the product of $|\mathcal{E}|$ and $|\V_R|$, resulting in an order-of-magnitude reduction in the overall number of required variables for prototypical problems.

\paragraph{Fractional output}
In order to adapt Problem \eqref{eq:P-AMoD-RT} for real-time control of AMoD systems, one last difficulty must be overcome.
The output of the problem is, in general, fractional: therefore it can not directly be used for control of individual vehicles. To overcome this, control actions are computed by \emph{sampling} the first time step of the fractional optimal solution to Problem \eqref{eq:P-AMoD-RT} in a receding-horizon framework.
In detail,
\begin{itemize}
\item \emph{Customer requests}: We recall that the intensity of a customer request $\lambda_m$ denotes the number of individual customers belonging to the request; each customer should be serviced by a separate vehicle. For each customer request $m$ departing at time $t=1$, the probability of selecting a vehicle with charge level $c$ to service customer request $m$ is set to $p_m(c)=\lambda_m^{c,\text{in}}/\sum_{\xi=1}^C \lambda_m^{\xi,\text{in}}$. One sample per customer is drawn from the distribution $p_m$, for a total of $\lceil \lambda_m\rceil $ samples.
Each customer is then assigned to a vehicle with a charge level corresponding to the sampled charge level. 
Analogously, for each  outstanding customer request $m^o$, the probability of departing at time $t$ and charge level $c$ is set to $p_{m^o}(t,c)=\lambda_{m^o}^{t,c,\text{in}}/\sum_{\tau=1}^T \sum_{\xi=1}^C \lambda_{m^o}^{\tau,\xi,\text{in}}$.
The number of samples drawn is $\lceil \lambda_{m^o} \rceil$, corresponding to the number of outstanding customers belonging to the request. Outstanding customers are assigned to a vehicle if the sampled departure time is $t=1$; in that case, the charge level of the vehicle corresponds to the sampled charge level. If no vehicles at the sampled charge level are available, a fall-back strategy is adopted where the customer is assigned to the closest vehicle with charge level  sufficient to complete the trip.

\item \emph{Idle vehicles}: Charging, discharging, and rebalancing actions are sampled from the distribution of the rebalancing flow $\{f_0(\mathbf{v},\mathbf{w})\}_{(\mathbf{v},\mathbf{w})}$. Specifically, for each node $v\in\V_R$ and each charge level $c\in\{1,\ldots, C\}$,
each edge $(\mathbf{v},\mathbf{w})$ with $\mathbf{v}=(v,c,1)$ is assigned a probability $p(\mathbf{v},\mathbf{w}) = f_0(\mathbf{v},\mathbf{w})/ \sum_{\mathbf{w'}:(\mathbf{v},\mathbf{w'})\in\mathcal{E}} f_0(\mathbf{v},\mathbf{w'} )$. One sample is drawn from $p(\mathbf{v},\mathbf{w})$ for each  vehicle charging, discharging, or rebalancing at node $v$ at charge $c$ and time $t=1$.
If the sampling procedure selects an edge corresponding to a charging link, a charging task is assigned to the vehicle; if an an edge corresponding to a discharging link is sampled, the vehicle is assigned a discharging task; if an edge corresponding to a road link is sampled, the vehicle is required to rebalance to the destination of the sampled road link.
\end{itemize}

The overall receding-horizon controller is presented in Algorithm \ref{alg:rt-sampling}.
\begin{algorithm}
\caption{Real-time receding-horizon algorithm for the P-AMoD problem}
\label{alg:rt-sampling}
\begin{algorithmic}
\Procedure{RHController}{Customer requests, vehicle states}
\State $\{f_0, \lambda^{c, \text{in}}_m, \lambda^{t,c, \text{out}}_m,\lambda^{t,c, \text{in}}_{m^o}, \lambda^{t,c, \text{out}}_{m^o}\}\gets $ Solve Problem \eqref{eq:P-AMoD-RT}
\ForAll{customer request $m$ with $t_m=1$}
	\State $p_m(c) \gets \lambda_m^{c,\text{in}}/\sum_{\xi=1}^C \lambda_m^{\xi,\text{in}} \quad \forall c\in\{1,\ldots, C\}$
	\State CustCharge($m$) $\gets $Sample $\lceil\lambda_m\rceil$ times from $\{p_m(c)\}_c$ 
\EndFor

\ForAll{outstanding customer request $m^o$}
	\State $p_{m^o}(t,c) \gets \lambda_{m^o}^{t,c,\text{in}}/\sum_{\tau=1}^T \sum_{\xi=1}^C \lambda_{m^o}^{\tau,\xi,\text{in}} \forall t\in\{1,\ldots, T\}, c\in\{1,\ldots, C\}$
	\ForAll {customers $\in\{1,\ldots,\lceil\lambda_m\rceil\}$}
	\State($t^o, c^o$) $\gets$ Sample from $\{p_{m^o}(t,c)\}_{t,c} $
	\If{$t^o=1$}
		\State CustCharge($m^o$) $\gets$ append $t^o$
	\EndIf	
	\EndFor
\EndFor
\ForAll{node $v\in\V_R$}
\ForAll{charge level $c\in[1,\ldots, C]$}
\State $r$ $\gets$ $\sum_{(\mathbf{v},\mathbf{w}): v_\mathbf{v}=v, t_\mathbf{v}=1,c_\mathbf{v}=c, }  f_0(\mathbf{v},\mathbf{w})$
\State $p(\mathbf{v},\mathbf{w}) \gets f_0(\mathbf{v},\mathbf{w})/ \sum_{\mathbf{w'}:(\mathbf{v},\mathbf{w'})\in\mathcal{E}} f_0(\mathbf{v},\mathbf{w'} )$ $\quad \forall  (\mathbf{v},\mathbf{w})\in\mathcal{E}$ such that $\mathbf{v}=(v,c,1)$
\For{$a=1,\ldots,r$}
\State $(\mathbf{v},\mathbf{w})$ $\gets$ Sample from $\{p(\mathbf{v},\mathbf{w})\}_{(\mathbf{v},\mathbf{w})}$
\If{$(\mathbf{v},\mathbf{w})$ is a charging link}
\State Task $\gets$ Charge at $v$
\ElsIf{$(\mathbf{v},\mathbf{w})$ is a discharging link}
\State Task $\gets$ Discharge at $v$
\ElsIf{$(\mathbf{v},\mathbf{w})$ is a road link}
\State Task $\gets$ Rebalance from $v_\mathbf{v}$ to $v_\mathbf{w}$
\EndIf
\State IdleTasks($v,c$) $\gets$ Append Task
\EndFor
\EndFor
\EndFor
\Return CustCharge, IdleTasks
\EndProcedure
\end{algorithmic}
\end{algorithm}

\subsubsection{Agent-based simulations}
\fi

\revTCNS{We assess the performance of  \ifshortest ~a receding-horizon implementation of the~ \else~ the receding-horizon~\fi \mbox{P-AMoD} controller with an agent-based simulation where a set of 1,257,916 individual commuting trips in Dallas-Fort Worth (based on \ifshortest~  American Communities Survey data\else~data from the American Communities Survey\fi) %
are serviced by an AMoD  fleet of 450,000 vehicles. The behavior of individual commuters and vehicles is tracked through an agent-based simulator.
Road congestion is modeled through the Bureau of Public Roads (BPR) model \cite{BPR1964}.
The receding-horizon P-AMoD controller does \emph{not} have access to the actual demand for transportation or for power; rather, the controller has access to noisy estimates corrupted by Gaussian noise.
The standard deviation of the transportation demand noise is 10\% of the mean (a very conservative figure compared to the performance of state-of-the-art tools for estimation of customer demand\ifshortest\else~\cite{IglesiasRossiEtAl2018}\fi) and the standard deviation of the exogenous power demand noise is 5\% of the mean.
Thus, the simulations characterize the behavior of the proposed P-AMoD controller in the presence of a high level of stochasticity in the demand for transportation and for power.
}

\ifshortest
\else
In the simulation, the generation costs are based on the \emph{marginal} cost of generation (from EIA estimates\ifextendedv~\cite[Table 8.4]{EIA2018}\fi)%
, to reproduce the strategic behavior of generator operators participating in a real-time electricity market. %
The uncoordinated controller may cause the power network to become unstable, causing the Economic Dispatch problem \eqref{eq:powerdispatch} to become infeasible. To account for this, 
we introduce slack variables in the power network balance equations \eqref{eq:powerbal} to capture the ISO's ability to disconnect loads to preserve the stability of the power network. The cost associated with the slack variables captures the economic loss borne by ISO users during a blackout (denoted as ``Value of Lost Load'' in the literature) and is set to \$6,000/MWh in accordance with ERCOT estimates\ifextendedv~\cite{LEI2013}\fi.
\fi

\ifshortest
\else
The receding-horizon problem is solved every 5 minutes with a 4-hour lookahead and a 15-minute time step.
\fi
The performance of the algorithm is  compared with a baseline case where no vehicles are present and an uncoordinated receding-horizon controller that optimizes the AMoD system's operations under the assumption that electricity prices stay constant.
\revTCNS{
 Table \ref{tab:results-real-time} shows the results.
\begin{table}[ht]
\centering
\caption{Real-time algorithm simulation results (10 hours).}
	\begin{tabular}{rrrrrrr}
	& \multicolumn{1}{c}{\!Baseline} & \multicolumn{1}{c}{P-AMoD} & \multicolumn{1}{c}{Uncoord.} \\ 
	\hline
	Avg. cust. travel time [h] & - & 1.594 & 1.559  \\
	Tot. energy demand [GWh] & 500.01 & 507.77 & 507.89  \\
		Blackouts [MWh] & 0 & 0 & 61.19\vspace{1mm}\\

	Tot. elec. expenditure, excl. TSO [k\$]& 15,067  & 15,067  & 17,845 \\
	Avg. price in DFW [\$/MWh] & 30.136 & 30.222 & 45.430 \\ %
	TSO tot. elec. expenditure [k\$] & - & 240.71  & 4,445.56 \\ %
		\end{tabular}
\label{tab:results-real-time}
\end{table}

\ifshortest
In absence of coordination, the AMoD system causes rolling blackouts in Dallas-Fort Worth: the Economic Dispatch problem  is infeasible for 79 of the  600 minutes considered in the simulation, an overall 61.19 MWh of power are not delivered to end users, and the average electricity price in Dallas-Fort Worth is  50\% higher than in the baseline case. The TSO's expenditure is over
18
times higher compared to the coordinated case.
Conversely, the P-AMoD system is able to ensure that the unit price of electricity in Dallas-Fort Worth only increases by 0.29\% compared to the case where no vehicles are present, despite the 4.84\% increase in power demand in the Dallas-Fort Worth region and the high level of uncertainty in the prediction of customer demand.  
\else
In absence of coordination, the AMoD system causes rolling blackouts in Dallas-Fort Worth: the Economic Dispatch problem  is infeasible for 79 of the  600 minutes considered in the simulation, and overall 61.19 MWh of power are not delivered to end users. The average electricity price in Dallas-Fort Worth is  \$45.43/MWh, 50\% higher than in the baseline case; across Texas,  the average price of electricity is 
\$43.23/MWh, and the total electricity expenditure for power network customers is over
16\%
higher compared to the case where no vehicles are present. The TSO's expenditure is over
18
times higher compared to the coordinated case.
Conversely, the P-AMoD system is able to ensure that the unit price of electricity (and therefore the expenditure of power network customers) in Dallas-Fort Worth and across Texas only increases by 0.29\% compared to the case where no vehicles are present, despite the 4.84\% increase in power demand in the Dallas-Fort Worth region and the high level of uncertainty in the prediction of customer demand. 
\fi
Thus, coordination between the AMoD system and the power network is vital to ensuring the stability of the power network. In absence of coordination, mass deployment of AMoD systems can heavily destabilize the power network, resulting in blackouts and excessive electricity prices; conversely, coordination is able to ensure that power prices remain virtually constant despite the increase in power demand, and is robust to large unmodelled stochastic variations in demand for transportation and for power.
}

\revTCNS{The receding-horizon P-AMoD problem was solved in an average of
61s
and a maximum of
162s; thus, the algorithm is amenable to closed-loop control of large-scale systems.}

\subsection{Proofs of all theorems}
\begin{proof}[Proof of Lemma \ref{lemma:bundledflows}]
The proof is constructive. 
 First we leverage the flow decomposition algorithm to decompose the bundled customer flow in a collection of path flows; next, we assign each path flow to a customer request; finally, we merge the path flows assigned to each request to obtain a feasible customer flow.
We assume without loss of generality that no two customer requests have the same origin node $v_m\in\mathcal{V}_R$, destination node $w_m\in\mathcal{V}_R$, and departure time $t_m\in\{1,\ldots, T\}$. Since customer routes are approximated as a network flow, if two or more such requests exist, they can be equivalently represented by a single request with intensity equal to the sum of the original requests' intensities.

 Define as path flow a network flow that has a fixed intensity on edges belonging to a path without cycles from the origin to the destination and zero otherwise. The flow decomposition algorithm \cite[Ch. 3.5]{AhujaMagnantiEtAl1993} can decompose the bundled customer flow into path flows. Specifically, the algorithm computes a collection of path flows $\mathcal{P}=\{f_p(\mathbf{u},\mathbf{v})\}_{p,(\mathbf{u},\mathbf{v})\in\mathcal{E}}$ such that, for every edge $(\mathbf{u},\mathbf{v})\in\mathcal{E}$, $\sum_p f_p(\mathbf{u},\mathbf{v}) = f_{B,d_B}(\mathbf{u},\mathbf{v})$. Each path flow $p\in\mathcal{P}$ has a single origin node $\mathbf{v}\in\mathcal{V} $ and destination node $\mathbf{w}\in\mathcal{V}$ with $v_\mathbf{w}=d_B$. %
Next, we assign each path flow to a customer request $(v_m,d_B,t_m,\lambda_m)$. %
Specifically, we decompose the path flows $\mathcal{P}$ in a collection of disjoint sets  $\{\mathcal{P}_{m}\}_m$ such that $\cup_{m=1}^M \mathcal{P}_m=\mathcal{P}$ and $\mathcal{P}_{m} \cap \mathcal{P}_{m'}=0$ for all $ m, m' \in \{1,\ldots, M\}$.
To do so, we assign all the path flows whose  origin node belongs to the set $\{\mathbf{v}=(v_m,t_m,c)\}_{c=1}^C$ to request $m$. By assumption, no two requests with the same destination $d_B$ can have the same origin location $v_m$ and departure time $t_m$: thus, every path flow is assigned to exactly one customer request $m$. The sum of the intensities of the path flows $p\in \mathcal{P}_m$ is $\lambda_m$; this property follows immediately from Equations \eqref{eq:bundledcustbal_conservation} and \eqref{eq:bundledcustbal_sumsource}. %
Finally, the customer flow for customer request $(v_m,d_B,t_m,\lambda_m)$ is obtained as the sum of the path flows in $\mathcal{P}_m$. By construction, each path flow satisfies Equation \eqref{eq:custbal:continuity}. Since the sum of the path flows equals $\lambda_m$, Equation \eqref{eq:custbal:intensity} is also satisfied by the sum of the path flows. This concludes the proof. 
\end{proof}

\begin{proof}[Proof of Theorem \ref{thm:equilibrium}]
The optimal solution to the P-AMoD problem also maximizes the revenue of the power generators  if locational marginal pricing is used \cite[Sec. 3]{WangNegrete-PinceticEtAl2012}. Thus, we focus on showing that the  optimal solution to the P-AMoD problem is also an optimal solution to the TSO's problem \eqref{eq:VRCP}.

The KKT stationarity conditions for the P-AMoD Problem \eqref{eq:P-AMoD}
for variables $\{f_\diamond, \lambda^{c, \text{in}}_m, \lambda^{t,c, \text{out}}_m, N_F\}$
are: 

\begingroup
{
\begin{subequations}
\begin{align}
&\frac{\partial(V_T T_M+V_D D_V + V_B)}{\partial f_\diamond(\mathbf{v},\mathbf{w})}+\transpose{\lambda^{\text{eq}}_\text{TSO}}\!\cdot\!\frac{\partial f_{\text{TSO}}^{\text{eq}}}{\partial f_\diamond(\mathbf{v},\mathbf{w})} + \transpose{\mu^{\text{ineq}}_\text{TSO}}\! \nonumber\\
& \cdot\!\frac{\partial f_{\text{TSO}}^{\text{ineq}}}{\partial f_\diamond(\mathbf{v},\mathbf{w})} +  \transpose{\lambda^{\text{eq}}_\text{ISO}}\!\cdot\!\frac{\partial f^{\text{eq}}_\text{ISO}}{\partial  f_\diamond(\mathbf{v},\mathbf{w})} +\transpose{\mu^{\text{ineq}}_\text{ISO}}\!\cdot\!\frac{\partial f_{\text{ISO}}^{\text{ineq}}}{\partial f_\diamond(\mathbf{v},\mathbf{w})}   \nonumber \\
& =0, \quad \forall \diamond \in \{0 \cup \{d_B \!\in\! \mathcal{D}\}\}, (\mathbf{v},\mathbf{w})\in\mathcal{E}, \label{eq:KKT-PAMoD-fm}\\
&\transpose{\lambda^{\text{eq}}_\text{TSO}}\cdot\frac{\partial  f_{\text{TSO}}^{\text{eq}}}{\partial \lambda^{c, \text{in}}_m} = 0, \quad  
 \forall c\in\{0,\ldots,C\}, m\in\{0, \ldots, M\}, \\
&\transpose{\lambda^{\text{eq}}_\text{TSO}} \cdot\frac{\partial f_{\text{TSO}}^{\text{eq}}}{\partial \lambda^{t,c, \text{out}}_m} = 0, \nonumber \\
& \quad 
 \forall c\in\{0,\ldots,C\}, t\in\{1,\ldots, T\}, m\in\{0, \ldots, M\}, \\
&\transpose{\lambda^{\text{eq}}_\text{TSO}} \cdot \frac{\partial f_{\text{TSO}}^{\text{eq}}}{\partial N_F(\mathbf{v})}=0, \quad \quad
 \forall \mathbf{v}\in\mathcal{V}.%
\end{align}
\label{eq:KKT-PAMoD}
\end{subequations}
}
\jvspace{-2mm}%
\endgroup

For a given set of variables $\{\theta^\star,p^\star \}$, the KKT conditions for Problem \eqref{eq:VRCP} are
{
\begin{subequations}
\begin{align}
&\frac{\partial(V_T T_M + V_D D_V + V_B )}{\partial f_\diamond(\mathbf{v},\mathbf{w})} \!+\! \frac{\partial(V_E)}{\partial f_\diamond(\mathbf{v},\mathbf{w})}\!+\!\transpose{\lambda^{\text{eq}}_\text{TSO}}\!\!\cdot\!\frac{\partial f_{\text{TSO}}^{\text{eq}}}{\partial f_\diamond(\mathbf{v},\mathbf{w})}  \nonumber \\
&+\!\transpose{\mu^{\text{ineq}}_\text{TSO}}\!\!\cdot\!\frac{\partial f_{\text{TSO}}^{\text{ineq}}}{\partial f_\diamond(\mathbf{v},\mathbf{w})}=0, 
{\forall \diamond \!\in\! \{0 \cup \{d_B \!\in\! \mathcal{D}\}\}, (\mathbf{v},\mathbf{w})\!\in\!\mathcal{E},\quad} \label{eq:KKT-VRCP-fm}\\
&\transpose{\lambda^{\text{eq}}_\text{TSO}}\cdot\frac{\partial  f_{\text{TSO}}^{\text{eq}}}{\partial \lambda^{c, \text{in}}_m} = 0, \quad 
 \forall c\in\{0,\ldots,C\}, m\in\{0, \ldots, M\}, \\
&\transpose{\lambda^{\text{eq}}_\text{TSO}} \cdot\frac{\partial f_{\text{TSO}}^{\text{eq}}}{\partial \lambda^{t,c, \text{out}}_m} = 0,\nonumber \\
& \quad 
 \forall c\in\{0,\ldots,C\}, t\in\{1,\ldots, T\}, m\in\{0, \ldots, M\},\\
&\transpose{\lambda^{\text{eq}}_\text{TSO}} \cdot\frac{\partial f_{\text{TSO}}^{\text{eq}}}{\partial N_F(\mathbf{v})}=0, \quad \quad
 \forall \mathbf{v}\in\mathcal{V}.
\end{align}
\end{subequations}
}
\jvspace{-3mm}

The second term in Equation \eqref{eq:KKT-VRCP-fm} is
{
\[
\frac{\partial(V_E)}{\partial f_\diamond(\mathbf{v},\mathbf{w})}=  1_{(\mathbf{v},\mathbf{w})\in\mathcal{E}_S} p_{(\mathbf{v},\mathbf{w})} \delta c_{v_\mathbf{v}},
\]
}
where $\delta c_{v_\mathbf{v}}=\delta c^+_{v_\mathbf{v}}$ if $c_\mathbf{w}>c_\mathbf{v}$ and  $\delta c_{v_\mathbf{v}}=\delta c^-_{v_\mathbf{v}}$ otherwise.

Leveraging Equation \eqref{eq:buscoupling}, the last two terms in Equation \eqref{eq:KKT-PAMoD-fm} can be rewritten as
{
\begin{align*}
&\transpose{\lambda^{\text{eq}}_\text{ISO}} \!\! \cdot \!\frac{\partial f^{\text{eq}}_\text{ISO}}{\partial  f_\diamond(\mathbf{v},\mathbf{w})} \!+\! \transpose{\mu^{\text{ineq}}_\text{ISO}}\!\!\cdot \! \frac{\partial f^{\text{ineq}}_\text{ISO}}{\partial  f_\diamond(\mathbf{v},\mathbf{w})} \!=\! 
 \sum_{l\in\mathcal{B}} \sum_{t=1}^T \Bigg[ \!\Big(\lambda^{\text{eq}}_\text{ISO}(l,t) 
    \\ &
  + \mu^{\text{ineq}}_\text{ISO}(l,t)\Big)
 \! \cdot \!
   \left( 1_{(\mathbf{v},\mathbf{w})\in M^+_{P,G}(l,t)} \!+\! 1_{(\mathbf{v},\mathbf{w})\in M^-_{P,G}(l,t)}\right) \Bigg]   J_C \delta c_{v_\mathbf{v}} 
.
\end{align*}
}

Every edge $(\mathbf{v},\mathbf{w})\in\mathcal{E}_S$ corresponds to a single load node $l\in\mathcal{B}: v_\mathbf{v}=\mathcal{M}_{\text{P,R}}(l)$ at a single time $t=t_\mathbf{v}$.Thus, the expression above can be rewritten as
{
\begin{align}
&\transpose{\lambda^{\text{eq}}_\text{ISO}}  \frac{\partial f^{\text{eq}}_\text{ISO}}{\partial  f_\diamond(\mathbf{v},\mathbf{w})} + \transpose{\mu^{\text{ineq}}_\text{ISO}}  \frac{\partial f^{\text{ineq}}_\text{ISO}}{\partial  f_\diamond(\mathbf{v},\mathbf{w})} \nonumber \\
& =  J_C \delta c_{v_w}\! \left(\lambda^{\text{eq}}_\text{ISO}(l_{v_{\mathbf{v}}},t_{\mathbf{v}}) +\mu^{\text{ineq}}_\text{ISO}(l_{v_{\mathbf{v}}},t_{\mathbf{v}}) \right), \label{eq:KKT-VRCP:dualprice}
\end{align}
}
where $ l_{v_{\mathbf{v}}}$ is such that $ v_\mathbf{v}=\mathcal{M}_{\text{P,R}}(l_{v_{\mathbf{v}}})$.

\revTCNSII{
Eq. \ref{eq:LMP} shows that the right-hand side of Eq. \ref{eq:KKT-VRCP:dualprice} equals the LMP at node $v_\mathbf{v}$.}
Therefore, Eq. \eqref{eq:KKT-VRCP-fm} and Eq. \eqref{eq:KKT-PAMoD-fm} are identical. As a result, the KKT conditions for the TSO's problem \eqref{eq:VRCP} are verified whenever the KKT conditions for the P-AMoD problem \eqref{eq:P-AMoD} are verified, and $\{f_\diamond^\star, \lambda^{c, \text{in}\star}_m, \lambda^{t,c, \text{out}\star}_m, N_F^\star\}$ is an optimal solution to Problem \eqref{eq:VRCP} for fixed $\{\theta^\star,p^\star \}$.

In conclusion, $\{f_\diamond^\star, \lambda^{c, \text{in}\star}_m, \lambda^{t,c, \text{out}\star}_m, N_F^\star\}$ is the solution to the TSO's Vehicle Routing and Charging Problem \eqref{eq:VRCP}  if the prices are set according to LMPs. In addition, the generation schedule $\{p^\star\}$ is the optimal (revenue-maximizing) schedule for self-interested power generators if the prices are set according to LMPs \cite[Sec. 3]{WangNegrete-PinceticEtAl2012}. That is, the set of variables $(\{f_\diamond^\star, \lambda^{c, \text{in}\star}_m, \lambda^{t,c, \text{out}\star}_m, N_F^\star\},\{\theta^\star\},\{p^\star \})$ is a general equilibrium for the P-AMoD market. This concludes the proof.
 \end{proof}

\end{appendix}
\jvspace{-2.3em}

\begin{IEEEbiography}[{\includegraphics[width=1in,height=1.2in,clip,keepaspectratio]{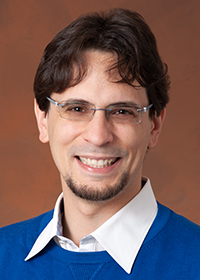}}]{Federico Rossi} is a Robotics Technologist at the Jet Propulsion Laboratory, California Institute of Technology.
He earned a Ph.D. in Aeronautics and Astronautics from Stanford University in 2018, a M.Sc. in Space Engineering from Politecnico di Milano and the Diploma from the Alta Scuola Politecnica in 2013.
His research focuses on optimal control and distributed decision-making in multi-agent robotic systems, with applications to robotic planetary exploration and coordination of fleets of self-driving vehicles for autonomous mobility-on-demand in urban environments.
\end{IEEEbiography}
\jvspace{-2.3em}
\begin{IEEEbiography}[{\includegraphics[width=1in,height=1.2in,clip,keepaspectratio]{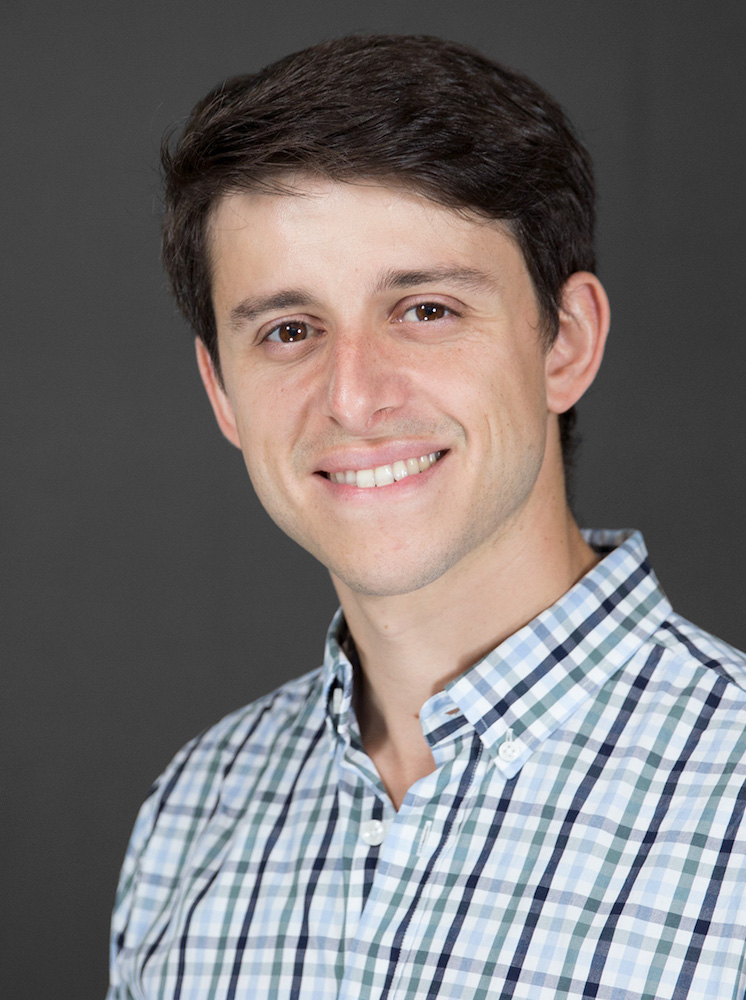}}]{Ramon Iglesias}  is a Ph.D. candidate in Civil and Environmental Engineering at Stanford University under the supervision of Marco Pavone. He develops algorithms and models to control large fleets of self-driving cars. More broadly, his research interests lie at the interplay between software systems and real-world infrastructure. Prior to his Ph.D, Ramon was a software engineer at SunPower. He has a M.S and a B.S. in Civil Engineering from Stanford and UT Austin, respectively.
\end{IEEEbiography}
\jvspace{-2.3em}
\begin{IEEEbiography}[{\includegraphics[width=1in,height=1.2in,clip,keepaspectratio]{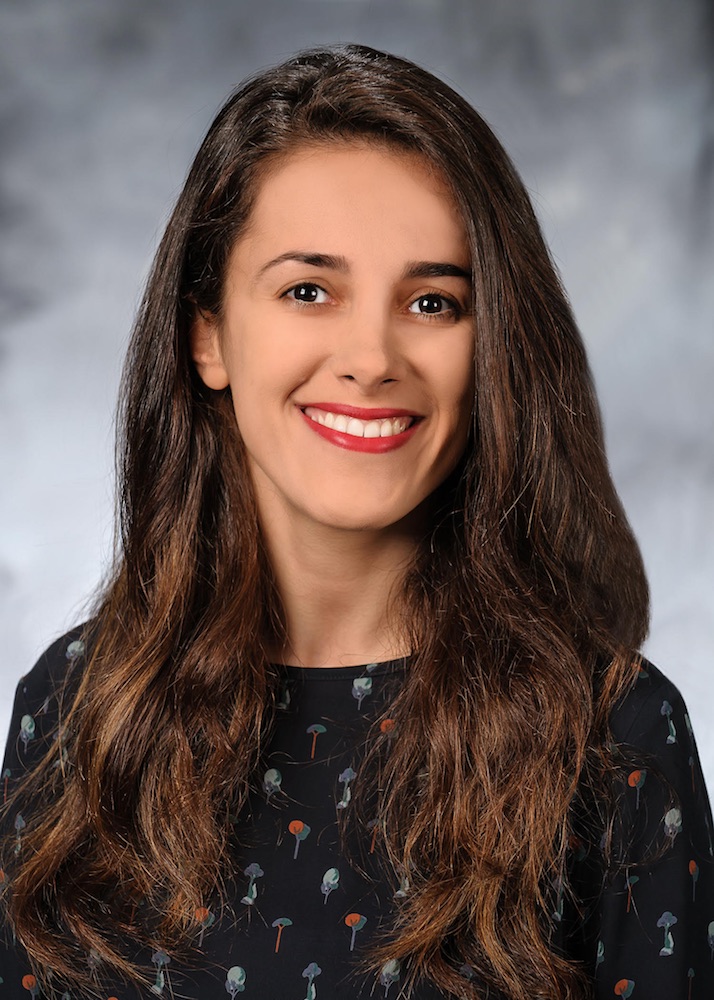}}]{Mahnoosh Alizadeh }  is an assistant professor of Electrical and Computer Engineering at the University of California Santa Barbara. 
Dr. Alizadeh received the B.Sc. degree in Electrical Engineering  from Sharif University of Technology in 2009 and the M.Sc. and Ph.D. degrees from the University of California Davis in 2013 and 2014 respectively, both in Electrical and Computer Engineering. From 2014 to 2016, she was a postdoctoral scholar at Stanford University. Her research interests are focused on designing scalable control and data analytic frameworks and market mechanisms for enabling sustainability and resiliency in societal infrastructure systems, with a particular focus on electric transportation systems. Dr. Alizadeh is a recipient of the NSF CAREER award.
\end{IEEEbiography}
\jvspace{-2.3em}
\begin{IEEEbiography}[{\includegraphics[width=1in,height=1.2in,clip,keepaspectratio]{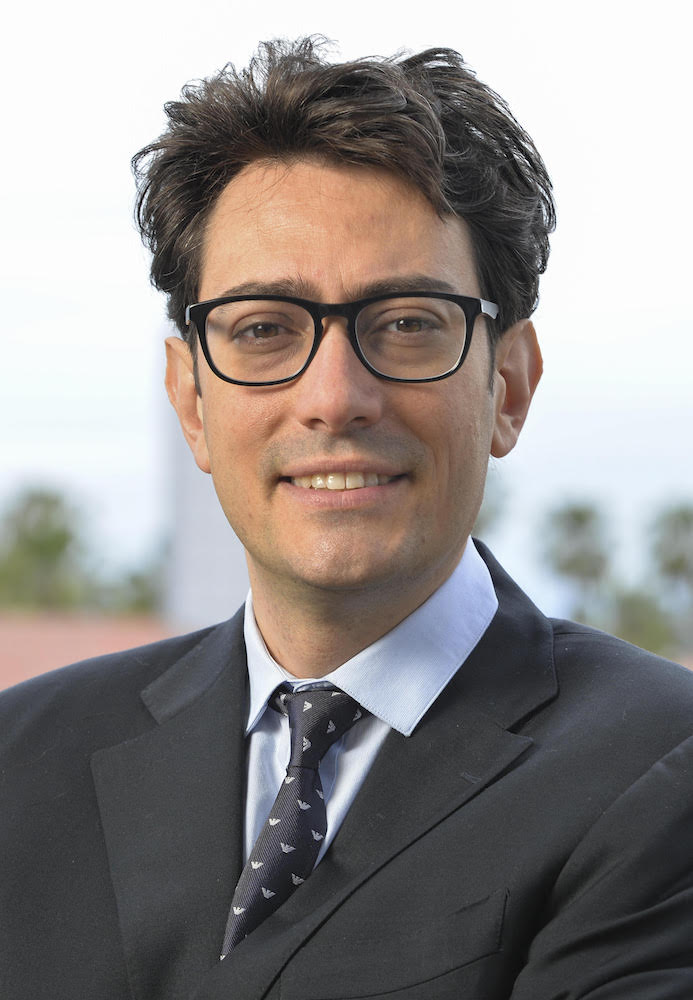}}]{Marco Pavone} is an Associate Professor of Aeronautics and Astronautics at Stanford University, where he is the Director of the Autonomous Systems Laboratory and Co-Director of the Center for Automotive Research at Stanford.
Before joining Stanford, he was a Research Technologist within the Robotics Section at the NASA Jet Propulsion Laboratory.
He received a Ph.D. degree in Aeronautics and Astronautics from the Massachusetts Institute of Technology in 2010. His main research interests are in the development of methodologies for the analysis, design, and control of autonomous systems, with an emphasis on self-driving cars, autonomous aerospace vehicles, and future mobility systems. He is a recipient of a number of awards, including a Presidential Early Career Award for Scientists and Engineers (PECASE), an ONR YIP Award, an NSF CAREER Award, and a NASA Early Career Faculty Award. He was identified by the American Society for Engineering Education (ASEE) as one of America's 20 most highly promising investigators under the age of 40. He is currently serving as an Associate Editor for the IEEE Control Systems Magazine.
\end{IEEEbiography}

\end{document}